\documentclass[a4paper,USenglish]{lipics}
\usepackage{todonotes}

\usepackage{microtype}%if unwanted, comment out or use option "draft"
\usepackage{proof}%if unwanted, comment out or use option "draft"
\usepackage{manfnt}
\usepackage{subfig}

\usepackage{amsmath}
\usepackage{comment}
\definecolor{light-gray}{gray}{0.86}

\newcommand{\mgray}[1]{\colorbox{light-gray}{$#1$}}
\newtheorem{proposition}{Proposition}

\title{Representing Nonterminating Reductions in $\mathbf{F}_2^\mu$}
\titlerunning{Representing Nonterminating Rewriting with $\mathbf{F}_2^\mu$} %optional, in case that the title is too long; the running title should fit into the top page column

\author[1]{Peng Fu}
\affil[1]{Dalhousie University \\
  \texttt{peng-fu@uiowa.edu}}
\authorrunning{P.\, Fu} %mandatory. First: Use abbreviated first/middle names. Second (only in severe cases): Use first author plus 'et. al.'

\Copyright{Peng Fu}%mandatory, please use full first names. LIPIcs license is "CC-BY";  http://creativecommons.org/licenses/by/3.0/

% mandatory: Please provide 1-5 keywords
\serieslogo{}%please provide filename (without suffix)
\volumeinfo%(easychair interface)
  {Billy Editor and Bill Editors}% editors
  {2}% number of editors: 1, 2, ....
  {Conference title on which this volume is based on}% event
  {1}% volume
  {1}% issue
  {1}% starting page number
\EventShortName{}
\begin{document}

\maketitle

\begin{abstract}
We specify a second-order type system $\mathbf{F}_2^\mu$ that is tailored for representing nonterminations. The nonterminating trace of a term $t$ in a rewrite system $\mathcal{R}$ corresponds to a productive inhabitant $e$ such that $\Gamma_{\mathcal{R}} \vdash e : t$ in $\mathbf{F}_2^\mu$, where $\Gamma_{\mathcal{R}}$ is the environment representing the rewrite system. 
  We prove that the productivity checking in $\mathbf{F}_2^\mu$ is decidable via a mapping to the $\lambda$-Y calculus.
We develop a type checking algorithm for $\mathbf{F}_2^\mu$ based on second-order matching. We implement the type checking algorithm in a proof-of-concept type checker. 
  \end{abstract}

\section{Introduction}
Nontermination has been an active research area in the term rewriting community. Early studies
includes classifying nonterminations based on the concept of looping reduction \cite{dershowitz1987}, i.e.
a reduction of the shape $t \to^+ C[\sigma t]$ for some substitution $\sigma$. More recently,
many nontermination detection techniques are proposed and implemented. Emmes et. al.~\cite{emmes2012proving} considered a generalized notion of looping reduction, e.g. $\sigma_2 \sigma_1^n t \to^+ C[\sigma_3 \sigma_2 \sigma_1^{f(n)} t]$ for some substitutions $\sigma_1, \sigma_2, \sigma_3$ and some ascending linear function $f$. Endrullis and Zantema~\cite{EndrullisZ15} used
a SAT solver to search for a non-empty regular language of terms such that it is closed under reduction and 
does not contain normal forms.

The nonterminating reductions are usually described using 
 mathematical notations and abbreviations. In this paper, we consider a novel representation using a relatively simple type system. 
 In particular, a nonterminating reduction of a term will be encoded as a proof evidence in a type system called $\mathbf{F}_2^\mu$. Representing nonterminating reduction is closely related to proving nontermination, but they have some subtle differences. For proving nontermination, it is enough to exibit a nonterminating reduction for
a term, while a term can admit multiple nonterminating reduction traces, with each trace exibits
a different kind of reduction pattern.

\begin{example}
  \label{fib}
Consider the following two string rewriting rules: $A \to_a AB, B \to_b A$. It is nonterminating by the observation that it contains the rule $A \to_a AB$, which means
there is a nonterminating reduction of the form $A \to_a AB \to_a ABB \to_a ABBB \to_a ...$. 
We can also use a L-system\footnote{See \url{https://en.wikipedia.org/wiki/L-system}.} like parallel reduction strategy to reduce $A$, this gives rise
to the nonterminating reduction: 
  ${{A}} \Longrightarrow {{A}B} \Longrightarrow {{A}BA} \Longrightarrow {{A}BAAB} \Longrightarrow
{{A}BAABABA} \Longrightarrow {{A}BAABABAABAAB} \Longrightarrow ...$. Note that
all the redexes at each step are reduced simultaneously and each word in the sequence is
a concatenation of the previous two. The aforementioned two reduction sequences are fundamentally different. The first one exibits a regular property, i.e. each string at each step can be described by the regular expression $AB^*$. In the second reduction sequence, each string is called a Fibonacci word, and the set of all such words is known to be \textit{context-free free}, i.e. any infinite subset can not be described by a context-free language \cite{yu2000properties}. We will
show how to represent the second reduction sequence in Section \ref{heuristic}.
\end{example}

%% In this paper, we show how to represent the nonterminating reductions of a term using the proof terms in a type system called $\mathbf{F}^\mu_2$, which is a subsystem of $\mathbf{F}_\omega$ with the fixed point typing rule. The idea is that we can represent a term rewriting system $\mathcal{R}$ as an typing environment $\Gamma_{\mathcal{R}}$, the first-order term $t$ as a type and the nonterminating reduction as a proof evidence $e$ such that $\Gamma_{\mathcal{R}} \vdash e : t$ in $\mathbf{F}^\mu_2$. %% The reduction patterns
%% of $t$ can be captured by its inhabitants and their typing derivations.    

The main contributions of the paper are the following ones.
\begin{itemize}
\item Inspired by Leibniz equality, we represent a rewrite rule $l \to r$ 
  as a typing environment $\kappa : \forall p . \forall \underline{x} . p\ r \Rightarrow p\ l$, where the type variable $p$ of kind $* \Rightarrow *$ represents a reduction context, $\kappa$
  is a fresh constant evidence and $\underline{x}$ denotes the set of variables in $l$. A specialized kind system is used to ensure the type variable of kind $* \Rightarrow *$ represents a reduction context. %% As a result, a reduction from $t_1$ to $t_2$ corresponds to
  %% the judegment $\Gamma_{\mathcal{R}} \vdash e : t_2 \Rightarrow t_1$, where $e$ contains the
  %% information of the reduction rules and contexts involved, and $\Gamma_{\mathcal{R}}$ is   an environment representing the rewrite system $\mathcal{R}$
  We call this representation of rewrite rule \textit{Leibniz representation} in Section \ref{leibniz-rep}.

  \item Nonterminating reductions would result in infinite proof evidence, we use
  the fixed point typing rule to represent the reductions finitely.
  Thus a nonterminating reduction of $t$ in $\mathcal{R}$ can be represented
  as $\Gamma_{\mathcal{R}} \vdash e : t$, where $e$ is an evidence containing a fixed point
  and $\Gamma_{\mathcal{R}}$ is the Leibniz representation of $\mathcal{R}$. 
  We called the resulting type system $\mathbf{F}_2^\mu$ (Section \ref{leibniz-rep}). 

\item We prove that if $\Gamma_{\mathcal{R}}\vdash e : t$ and $e$ is \textit{hereditary head normalizing}(HHN), then we
  can recover from the evidence $e$ a nonterminating reduction of $t$ (Section \ref{meta}). We also prove that the hereditary head normalization is decidable in $\mathbf{F}_2^\mu$. The decidability result is obtained via a mapping from $\mathbf{F}_2^\mu$ to $\lambda$-Y
  calculus, for which HHN is decidable. 
  
\item It is more convenient to write the unannotated proof evidence and let the type checker
  fill in the annotations. For this purpose we develop a second-order
  type checking algorithm in Section \ref{typecheck} and Section \ref{heuristic}. It simplifies the process of representing nonterminations
  in $\mathbf{F}_2^\mu$. We implement a prototype type checker\footnote{The prototype type checker is available at \url{https://github.com/Fermat/FCR}} based on this algorithm and
  give some nontrivial examples in the Appendix. 
\end{itemize}

All the examples and the missing proofs in this paper may be found in the Appendix. 

\section{The Main Idea}
\label{idea}
First, let us consider how to represent a rewrite system in a type system. 
We could model the rewrite rule $l \to r$ as a typing environment
$\kappa : l \Rightarrow r$, like many proof systems for rewriting (\cite{bezem2003term}, \cite{stump2005logical}). 
However, modeling the rewrite rule $l \to r$ as an implication type
$l \Rightarrow r$ will make it difficult to observe the proof evidence.
For example, suppose we have 
a set of ground rewrite rules $A_i \to A_{i+1}$  modelled by $\kappa_i : A_i \Rightarrow A_{i+1}$ for $0 \leq i \leq n$ for some $n$, where $\kappa_i$ is a constant. Then the evidence for
the reduction $A_0 \to^* A_n$ would be $\lambda \alpha . (\kappa_n\ ...\ (\kappa_0\ \alpha)\ ...) : A_0 \Rightarrow A_n$. Informally, we can see
that the evidence $\lambda \alpha . (\kappa_n\ ...\ (\kappa_0\ \alpha)\ ...)$ grows outward as the number $n$ gets larger. 
When the reduction is nonterminating, it would be difficult to observe the very first step of the reduction ($\kappa_0$). Fortunately, this difficulty can be overcome by representing $l \to r$ as $r \Rightarrow l$.
Thus we have the evidence $\lambda \alpha . (\kappa_0\ ...\ (\kappa_n\ \alpha)\ ...) : A_n \Rightarrow A_0$, with $\kappa_i : A_{i+1} \Rightarrow A_i$ for all $0 \leq i \leq n$.
So we can easily observe the first step of the reduction $\kappa_0$ at the outermost position.

Next, we need to model the reduction context in rewriting. Given a rewrite rule $l \to r$,
we have a one-step reduction $C[l] \to C[r]$ for any first-order term context $C$. 
Inspired by \textit{Leibniz equality}, 
we use the type $\forall p . p\ r \Rightarrow p\ l$ to model the rewrite rule $l \to r$. 
The intended reading for this type is that $l$ can be replaced by $r$ under \textit{any} 
first-order term context $p$. Note that $p$ is a second-order type variable of kind $* \Rightarrow *$. So we can
 obtain $C[r] \Rightarrow C[l]$ by instantiating $p$ with $\lambda x . C[x]$ in $\forall p . p\ r \Rightarrow p\ l$. This motivates our definition of \textit{Leibniz representation} for the rewrite rules in Section \ref{leibniz-rep} and the use of the type system $\mathbf{F}_2^\mu$, as its kind
 system %is a specialization of the kinding of $\mathbf{F}_{\omega}$ so that one
 enforces that one can only instantiate type variable of kind $* \Rightarrow *$ with a type that represents a term context.

 Last but not least, we need a mechanism to handle the nonterminating reductions.
Consider the cycling rewrite rules: $A \to B$ and $B \to A$,
which are represented as two axioms $\Gamma = \kappa_A : B \Rightarrow A, \kappa_B : A \Rightarrow B$.
There is a cyclic reduction for $A$:  $A \to B \to A \to B \to ...$. Using
the Leibniz representation, the corresponding proof evidence for this
reduction would be an infinite proof evidence $\kappa_A \ (\kappa_B\ (\kappa_A\ (\kappa_B \ ...\ )))$. But we want to use a finite evidence $e$ to represent this nonterminating reduction. The solution here is to use a fixpoint operator. We can represent the infinite proof evidence
finitely as $\mu \alpha. \kappa_A\ (\kappa_B \ \alpha)$, where the $\mu$ is a fixpoint binder
with the operational meaning of $\mu \alpha.e \leadsto [\mu \alpha.e/\alpha]e$. 
This motivates the following fixed point typing rule for $\mathbf{F}_2^\mu$.

\begin{center}
  \begin{tabular}{l}
    \infer[\textit{Mu}]{\Gamma \vdash \mu \alpha .e : T}{\Gamma, \alpha : T \vdash e : T}
  \end{tabular}
\end{center}

\noindent So $\Gamma \vdash \mu \alpha. \kappa_A\ (\kappa_B \ \alpha) : A$ represents a nonterminating reduction of the shape $A \to B \to A \to B \to ...$, since the
unfolding of the evidence $\mu \alpha. \kappa_A\ (\kappa_B \ \alpha)$ gives the
sequence of rules that we are going to apply. Note that not all evidence of type $A$
are representing nonterminating reductions. For example, according
to the typing rule \textit{Mu}, we have $\Gamma \vdash \mu \alpha . \alpha : A$, but $\mu \alpha. \alpha$ does not give any information to reconstruct a nonterminating reduction. We
show in Section \ref{meta} that only the \textit{hereditary head normalizing} evidence are representing the nonterminating reductions.

We conclude this section by recasting our idea in the following example.

\begin{example}
  \label{gg}
 Consider the following rewrite rule.
  \begin{center} $F\ x \to G \ (F \ (G\ x))$
    \end{center}
The term $F \ x$ admits a reduction
sequence $F\ x \to G\ (F \ (G\ x)) \to G^2\ (F \ (G^2\ x)) \to G^3\ (F \ (G^3\ x)) \to ...$,
where $G^i\ x$ is a shorthand for $\underbrace{G\ (G\ ... (G}_i\ x)...)$ for any $i>1$. Using the Leibniz representation, the rewrite system is represented by the following $\mathbf{F}_2^\mu$ environments:

  \begin{center}
   {
     $ \Delta = F : * \Rightarrow *, G : * \Rightarrow *$ \\
     $\Gamma =  \kappa : \forall p . \forall x . p\ (G \ (F \ (G\ x))) \Rightarrow p \ (F\ x)$}
 \end{center}
\noindent  Note that $\kappa : \forall p . \forall x . p\ (G \ (F \ (G\ x))) \Rightarrow p \ (F\ x)$ corresponds to 
  the rewrite rule $F\ x \to G \ (F \ (G\ x))$, where $p$ of kind $* \Rightarrow *$ corresponds to a reduction context.

We will first construct a hereditary head normalizing (productive) evidence $e$ such that $\Gamma\vdash e : F \ x$. Then we will show how to check whether such $e$ is indeed representing the nonterminating reduction above. It is enough to derive $\Gamma \vdash e' : \forall p. \forall x. p\ (F \ x)$ for some $e'$. Consider the following judgement. 

\begin{center}
(1)  $\Gamma , \alpha : \forall p . \forall x. p \ (F \ x) \vdash \lambda
  p . \lambda x . \alpha\ (\lambda y . p\ (G\ y))\ (G\ x) : \forall p
  . \forall x . p\ (G \ (F \ (G\ x)))$

\end{center}

%\knote{check below?}
In (1), we instantiate the type of $\alpha$ as follows: % $\alpha\ (\lambda y . p\ (G\ y))\ (G\ x)$  instantiates $p$ and $x$
 $p$ is instantiated by $\lambda y . p\ (G\ y)$ and $x$ is instantiated by $G\ x$. Since
 we know that $(\lambda y . p\ (G\ y))\ (F \ (G\ x)) = p\ (G\ (F\ (G\ x)))$, thus
 $\alpha\ (\lambda y . p\ (G\ y))\ (G\ x)$ has the type 
 $p\ (G\ (F\ (G\ x)))$. The lambda-abstractions $\lambda p.\lambda x. $ is used to quantify over $p$ and $x$ in the type of $\alpha\ (\lambda y . p\ (G\ y))\ (G\ x)$.

 From $\forall p . \forall x . p\ (G \ (F \ (G\ x))) \Rightarrow p \ (F\
   x)$ and $\forall p
  . \forall x . p\ (G \ (F \ (G\ x)))$, we can deduce the following. 

\begin{center}
(2)  $\Gamma , \alpha : \forall p . \forall x. p \ (F \ x) \vdash \lambda
  p . \lambda x . \kappa\ p\ x\ (\alpha\ (\lambda y . p\ (G\ y))\ (G\
  x)) : \forall p. \forall x . p \ (F \ x)$
\end{center}

\noindent We now can apply \textit{Mu} rule to (2) and obtain the following: 

\begin{center}
(3)  $\Gamma \vdash e' \equiv \mu \alpha . \lambda p . \lambda x . \kappa\ p\ x\
  (\alpha\ (\lambda y . p\ (G\ y))\ (G\ x)) : \forall p. \forall x . p
  \ (F \ x)$
\end{center}

\noindent Thus by instantiation we have $\Gamma \vdash e' \ (\lambda y.y)\ x : F\ x$.
Observe the following unfolding of $e' \ (\lambda y.y)\ x$ (we use beta-reduction and $\mu \alpha . e \leadsto [\mu \alpha . e/\alpha]e$ to perform reduction):

\begin{center}
\footnotesize{$e' \ (\lambda y.y)\ x  \leadsto^* \mgray{\kappa\ (\lambda y.y)\ x}\ (e'\
  (\lambda y . G\ y)\ (G\ x))  \leadsto^*
  \kappa\ (\lambda y.y)\ x \ (\mgray{\kappa\ (\lambda y . G\ y)\ (G\ x)} \ (e'\
  (\lambda y . G \ (G\ y))\ (G\ (G\ x)))) \leadsto^* ...$}
\end{center}

\noindent As $\kappa$ takes a reduction context and an instantiation as its first two arguments, the gray subterms $\kappa\ (\lambda y.y)\ x$ and $\kappa\ (\lambda y . G\ y)\ (G\ x)$ can be read as: the first step
of the reduction for $F\ x$ is under the empty context $\bullet$ using $\kappa$ with the instantation $[x/x]$. The second step is also using the $\kappa$ rule, reducing the redex under the term context $G\ \bullet$, with the instantiation $[G\ x/x]$. As $e'\ (\lambda y.y) \ x$ is
 hereditary head normalizing (productive), the exact reduction information for $F\ x$ can be obtained from the unfolding. 

With the help of the prototype type checker for $\mathbf{F}_2^\mu$, the construction of the fully annotated evidence $e' \ (\lambda y.y)\ x$ can be semi-automated. For this example, the user will need to provide the following. 

{\footnotesize
\begin{verbatim}
K : forall p x . p (G (F (G x))) => p (F x)
h : forall p x . p (F x)
h = K h
e : F x 
e = h 
\end{verbatim}
}

 %Note that the declaration \texttt{K} corresponds to the rewrite rule.
%we declared our intent to proceed coinductively, 
\noindent The corecursive equation \texttt{h = K h} can be viewed as a proof sketch for

\noindent {\texttt{forall p x . p (F x)}}, it reflects the observation that the rule \texttt{K} is repeatedly applied in 
the reduction for \texttt{F x}.
The declaration \texttt{e : F x = h}
means that in this case we are providing an evidence for the nonterminating reduction of the term
\texttt{F x} under the empty term context. The type checker will try to fill in the exact term contexts and instantiations using the type checking algorithm we developed. It gives the following output (no existing first-order type checking algorithm can type check the above code).

{\footnotesize
\begin{verbatim}
e : F x = h (\ x1' . x1') x
h : forall p x . p (F x) = 
     \ p0' x1'. K (\ m1' . p0' m1') x1'
                  (h (\ m1' . p0' (G m1')) (G x1'))
\end{verbatim}
}
\end{example}

\section{Modeling First-order Term Rewriting System in $\mathbf{F}_{2}^{\mu}$}
\label{leibniz-rep}
%\section{Preliminary}
To model term rewriting, we define the type system $\mathbf{F}_{2}^{\mu}$, which restricts the type abstraction of
$\mathbf{F}_{\omega}$ \cite{Girard:72} to second-order. We define Leibniz representation of rewrite rules (Definition \ref{leibniz}) and show how
it can model rewriting via Theorem \ref{ade}. 

\begin{definition}[Syntax of $\mathbf{F}_{2}^{\mu}$]
  \label{syntax}

  \

  {\footnotesize
  \begin{tabular}{llll}
    \textit{Evidence} & $e$ & $::=$  & $\alpha \ | \ \kappa ~\mid~ \lambda \alpha . e ~\mid~ e \ e' ~\mid~
    \mu \alpha . e \ | \ e\ T \ | \ \lambda x. e$
    \\
   \textit{Term Kinds} & $K$ &  $::=$ & $ * \ | \ * \Rightarrow K $
   \\
    \textit{Kinds}& $k$ & $::=$ & $o \ | \ K $
    \\
    \textit{Types} & $T$ & $::=$ & $F \ | \ x \ | \ \lambda x . T \ | \  \forall x : K . T \ | \ T\ T' \ |\ T \Rightarrow T'$
    \\
    \textit{Environment} & $\Gamma$ & $::=$ & $\cdot ~\mid~ \alpha : T, \Gamma \ | \ \kappa : T, \Gamma$
   \\
    \textit{Type Environment} & $\Delta$ & $::=$ & $\cdot ~\mid~ x : K, \Delta \ | \ F : K, \Delta$

  \end{tabular}}
\end{definition}

Note that $\kappa$ denotes an evidence constant and is used to label rewrite rules (see Definition \ref{leibniz}). The letters such as $F, G$ is used to denote constant types. 
We use letters such as $\alpha, \beta$ to denote evidence variables, and $x, y$ to denote type
variables.
We use $\lambda x.e$ to denote type-abstraction on the evidence. Fixed point abstraction $\mu$ in $\mu \alpha . e$ binds the variable $\alpha$ in $e$. Operationally, $\mu \alpha . e$ behaves in the same was as the lambda term
$\mathbf{Y}\ (\lambda \alpha .e)$, where $\mathbf{Y}$ is a fixpoint combinator. In our paper $\mu f . \lambda \alpha_1 .... \lambda \alpha_n . e$
is also represented by the corecursive equation $f \ \alpha_1 \ ... \ \alpha_n = e$.
%\knote{the above does not look like a recursive equation...}
% We use $\mu \alpha . e$ for an easy integration with typing rule.
We use $\forall \underline{x}.T$ as a shorthand for $\forall x_1. ... \forall x_n . T$, and $e\ \underline{e'}$ for $e \ e_1'\ ... \ e_n'$, where the number $n$ is not important.  

 We distinguish two notions of kinds: kind $o$ is intended to classify types that are of formula nature, while kind $K$ is intended to classify types that are of first-order term nature.
Observe that we only allow quantification over the variables of kind $K$ for a type. 
We use $*^n \Rightarrow *$ as a shorthand for $\underbrace{* \Rightarrow ... \Rightarrow *}_n \Rightarrow * $. 
%Intuitively, $F :  *^n \Rightarrow *$ means $F$ is a function symbol of arity $n$. 

%We will further motivate our formulation of kinds later in this section.

Comparing to $\mathbf{F}_{\omega}$, 
the following kinding rules of $\mathbf{F}_2^\mu$ restrict the level of type abstraction
 to second-order, and stratify the types into two kinds.%% \footnote{The standard kinding rules for $\mathbf{F}_\omega$ are in the Appendix file.}.

\begin{definition}[Kinding Rules]
\label{kindsystem}
\fbox{$\Delta \vdash T : k$}
\

{\footnotesize
\begin{tabular}{lll}
\\
\infer{\Delta \vdash x | F : K}{(x | F :  K) \in \Delta}    
&

\infer{\Delta \vdash T_2\ T_1 : K}{\Delta \vdash T_1 : * & \Delta \vdash T_2 : * \Rightarrow K}
&

\infer{\Delta \vdash \lambda x. T : * \Rightarrow K}{\Delta, x : * \vdash T : K & x \in \mathrm{FV}(T)}

\\
\\

\infer{\Delta \vdash  \forall x : K . T : o }{\Delta, x : K \vdash T : o | * } %\phi(T)

&

\infer{\Delta \vdash  T \Rightarrow T' : o }{\Delta \vdash T : o | * & \Delta \vdash T' : o | *}

  \end{tabular}  
}

\end{definition}
We use $(x | F :  K) \in \Delta$ to abbreviate $x : K \in \Delta$ or $F : K \in \Delta$.
And $\Delta \vdash T : o | *$ means $\Delta \vdash T : o$ or $\Delta \vdash T : *$.
The kinding rule for $\lambda x.T$ is \textit{relevant}, i.e. the lambda bound variable $x$ must
be used in $T$. We have this requirement is because we want types of kind $* \Rightarrow *$ to represent a first-order term context with at least a hole, as the proof of Theorem \ref{faithfulness} needs this.
Given an environment $\Delta$, it is decidable whether a type $T$ is well-kinded. Given
a type $T$, it is also decidable to check if there is a $\Delta$ such that $\Delta \vdash T : k$ for some kind $k$. We use $\forall x . T$ instead of $\forall x : K . T$ in our examples.  
The kind system allows us to separate two different kinds of types in $\mathbf{F}_{2}^{\mu}$:  types that will be used to represent first-order 
terms and types that allow variable instantiation and modus ponens.
%\knote{again formulas}
 % Our kinding system ensures these two roles. 

\begin{definition}
  We define a reduction relation $T \to_o T'$ on types, it is the compatible closure of type level beta reduction
  $(\lambda x . T)\ T' \to_o [T'/x]T$. 
\end{definition}
\begin{proposition}
  If $\Delta \vdash T : k$, then $T$ is strongly normalizing with respect to $\to_o$, and $\to_o$
  is confluent. 
\end{proposition}

Let $\mathrm{FV}(T)$ denote the set of free variables occuring in $T$. The following theorem shows that the kind system satisfies the subject reduction property and the set of free type variables is invariant
under the $\to_o$-reduction.

\begin{theorem}[Subject Reduction for Kinding]
  \label{sub:kind}
 If $\Delta \vdash T : k$ and $T \to_o T'$, then $\mathrm{FV}(T) = \mathrm{FV}(T')$ and $\Delta \vdash T' : k$. 
\end{theorem}
\begin{definition}[Second-order Types]
  A type $T$ is \textit{flat} iff it is one of the following forms: (1) $T \equiv x$ or $T \equiv F$. (2) $T \equiv T_1 \ T_2$, where $T_1, T_2$ are flat. We say a type $T$ is \textit{second-order} if $T$ is flat or $T \equiv \lambda x_1. ...\lambda x_n . T'$, where $T'$ is flat and $x_i \in \mathrm{FV}(T')$ forall $x_i \in \{ x_1, ..., x_n\}$.
\end{definition}

Note that types such as $\lambda x . F \ x \ x$, $\lambda x. \lambda y. F \ x \ y, \lambda x.x$ are second-order, but $\lambda x .\lambda y. F\ x \Rightarrow F\ y$ are not second-order. We use second-order types to model both first-order term contexts and terms. The following theorem shows that the kind system stratifies types into two kinds. 

\begin{theorem}[Properties of Kinding]
  \label{prop:kinding}
  \
  
  \begin{enumerate}
  \item If $\Delta \vdash T : o$, then $T$ is of the form $\forall x. T'$ or $T_1 \Rightarrow T_2$.
    \item If $\Delta \vdash T : *^n \Rightarrow *$, then the $\to_o$-normal form of $T$ is second-order.
  \end{enumerate}
\end{theorem}

We define reduction rules for the evidence in the following.

\begin{definition}[Evidence Reduction]
  \label{headred} 
\

{\footnotesize
\textit{Head reduction context}  $\mathcal{H} \  ::=  \ \bullet \ | \ \mathcal{H}\ e \ | \ \lambda \alpha. \mathcal{H} \ | \ \lambda x . \mathcal{H}$

\textit{General reduction context} 
$\mathcal{C} \  ::=  \ \bullet \ | \ \mathcal{C}\ e \ | \ \mathcal{C}\ T  \ | \ \lambda \alpha. \mathcal{C} \ | \ \lambda x . \mathcal{C} \ | \ e \ \mathcal{C} \ | \ \mu \alpha. \mathcal{C} $
}

\

{\footnotesize
%\[\begin{array}{c}
$\mathcal{H}[\mu \alpha . e] \leadsto_h \mathcal{H}[[\mu \alpha . e/\alpha]e]
\quad
\mathcal{H}[(\lambda \alpha . e) \ e'] \leadsto_h \mathcal{H}[[e'/\alpha]e]
\quad
\mathcal{C}[(\lambda x . e) \ T] \leadsto_\tau \mathcal{C}[[T/x] e]$

\

$\mathcal{C}[\mu \alpha . e] \leadsto_\mu \mathcal{C}[[\mu \alpha . e/\alpha]e]
\quad \mathcal{C}[(\lambda \alpha . e) \ e'] \leadsto_\beta \mathcal{C}[[e'/\alpha] e]$
\quad $\mathcal{C}[T] \leadsto_o \mathcal{C}[T']$ if $T \to_o T'$

}
%  \end{array}
%\]}

\end{definition}

%% We require all the types in the evidence are in $\to_o$-normal form when performing the evidence reduction.
We call the one step reduction $\leadsto_h \cup \leadsto_\tau \cup \leadsto_o$ a one step \textit{head reduction}\footnote{This definition is following Barendregt \cite{barendregt1984lambda}, Page 173.}, denoted by $\leadsto_{h\tau o}$. The head reduction is lazy, i.e., $\mu \alpha . \kappa \ \alpha$ is normalizing with head reduction. We call an 
evidence a \textit{head normal form} if it can not be one step reduced by $\leadsto_{h\tau o}$.

\begin{theorem}
  \label{confluence}
 $\leadsto_{\beta\mu\tau o}$ and
    $\leadsto_{h\tau o}$ are confluent, and $\leadsto_\tau$ is strongly normalizing.  
\end{theorem}

We specify the typing rules for $\mathbf{F}_{2}^{\mu}$ in the following.

\begin{definition}[Typing of $\mathbf{F}_{2}^{\mu}$]
\label{proofsystem}
\

{\footnotesize
\begin{tabular}{lll}
\\
\infer{\Gamma \vdash \alpha | \kappa : T}{(\alpha | \kappa :  T) \in \Gamma}    
&

\infer[(\textsc{App})]{\Gamma \vdash e_2\ e_1 : T}{\Gamma \vdash e_1 : T' & \Gamma \vdash e_2 : T' \Rightarrow T}
&
\infer[(\textsc{Lam})]{\Gamma \vdash \lambda \alpha.e : T' \Rightarrow T}{\Gamma, \alpha : T' \vdash e : T}
\\
\\
\infer[(\textsc{Mu})]{\Gamma \vdash \mu \alpha . e : T}{\Gamma, \alpha : T \vdash e : T}
&
\infer[(\textsc{Inst})]{\Gamma \vdash e \ T': [T'/x]T}{\Gamma \vdash e : \forall x : K . T}

&
\infer[(\textsc{Abs})]{\Gamma \vdash \lambda x . e: \forall x : K . T}{\Gamma \vdash e :  T & x \notin \mathrm{FV}(\Gamma)}

\\
\\

\infer[(\textsc{Conv})]{\Gamma \vdash e : T'}{\Gamma \vdash e : T & T \leftrightarrow^*_o T'}

&
%\infer[(\textsc{Beta})]{(\lambda x . T) \ T' = [T'/x]T}{}

  \end{tabular}  
}
\end{definition}
%$x \notin \mathrm{FV}(\Gamma)$ means $x$ does not occur free in $\Gamma$.
%% Note that $(\alpha | \kappa :  T) \in \Gamma$ means $\alpha : T \in \Gamma$ or $\kappa : T \in \Gamma$.
In the \textsc{Abs} rule, only the types of kind $K$ are quantified. We use $\mathrm{FV}(\Gamma)$
to denote the set of free type variables occurs in $\Gamma$.
We require that all the types are well-kinded. Since $\to_o$ is strongly normalizing and confluent,
we will work with types in $\to_o$-normal form in this paper. The rule \textsc{Conv} is used implicitly. %% To achieve decidable type checking, we 
%% need to further annotate the lambda-binder in the \textsc{Lam} rule, which is standard.
%we elide such treatment
%in the paper.
%\knote{the above para required several little corrections, please check it is ok}

The followings theorems shows that the type system $\mathbf{F}_{2}^{\mu}$ has the usual inversion
and subject reduction properties.

\begin{theorem}[Selected Inversion Theorems]
\label{thm:inversion}
\

  \begin{enumerate}
  %% \item If $\Gamma \vdash x : T$, then there exists $(x : T') \in \Gamma$ and $T \leftrightarrow_o^* T'$.
  %% \item If $\Gamma \vdash \kappa : T$, then there exists $(\kappa : T') \in \Gamma$ and $T \leftrightarrow_o^* T'$.
%% \item If $\Gamma \vdash \lambda \alpha . e : T$, then $\Gamma, \alpha : T_1 \vdash e : T_2$ and $T_1 \Rightarrow T_2 \leftrightarrow_o^* T$.

\item  If $\Gamma \vdash e\ e' : T$, then $\Gamma \vdash e : T_1 \Rightarrow T_2$, $\Gamma \vdash e' : T_1$
  and $T_2 \leftrightarrow_o^* T$.

%% \item If $\Gamma \vdash \lambda x . e : T$, then $\Gamma, x : K \vdash e : T'$ and $\forall x. T' \leftrightarrow_o^* T$.

\item  If $\Gamma \vdash e\ T_1 : T$, then $\Gamma \vdash e : \forall x : K. T'$ and $[T_1/x]T' \leftrightarrow_o^* T$.

%% \item If $\Gamma \vdash \mu \alpha . e : T$, then $\Gamma, \alpha : T' \vdash e : T'$ and $ T' \leftrightarrow_o^* T$.
  \end{enumerate}
\end{theorem}

\begin{theorem}[Subject Reduction]
\label{sub-red}
  If $\Gamma \vdash e : T$ and $e \leadsto_{h\tau o} e'$, then $\Gamma \vdash e' : T$. 
\end{theorem}

Due to \textsc{Mu} rule, $\mathbf{F}_{2}^{\mu}$ allows diverging evidence with respect to $\leadsto_{\beta\mu}$. We will focus on the \textit{hereditary head normalizing} evidence (Definition \ref{here}),
which will be discussed in Section~\ref{meta}. 

\begin{definition}[Terms and Contexts]
  \label{fo-term}
\

{
\textit{First-order term} $t, l, r \ ::= \ x \ | \ F^{n} \ t_1 \ ...\ t_n$

\textit{Term context} $C \ ::= \ \bullet \ | \ x \ |\ F^{n} \ C_1 \ ...\ C_n$
}
\end{definition}

Note that the term context can contains multiple $\bullet$ and we use the 
the notation $C[t_1,...,t_n]$ to denote the result of replacing $\bullet$ from left to
right in $C$ by $t_1,..., t_n$. A special case is $C[t]$, it means there is exactly one
$\bullet$ in $C$, which is replaced by $t$. 
The function symbol $F$ of arity $n$ is denoted by $F^n$.
We work with applicative first-order terms in this paper, and we assume all function symbols are fully applied, thus we often write
$F \ t_1 \ ...\ t_n$ instead of $F^n \ t_1 \ ...\ t_n$.
We reuse $\mathrm{FV}(t)$ to mean the set of free variables in $t$. 
\begin{definition}[Rewrite Rules]
Suppose $l$ and $r$ are first-order terms, where $l$ is
 not a variable and $\mathrm{FV}(r) \subseteq \mathrm{FV}(l)$, then  $l \to r$  is a first-order rewrite rule. 
A rewriting system is a set  $\mathcal{R}$ of rewrite rules.
 We write $C[t] \to C[t']$ if
  there exists $l \to r \in \mathcal{R}$ such that $\sigma l \equiv t$ and $\sigma r \equiv t'$ for some substitution $\sigma$.
\end{definition}

\noindent \textbf{Important Notation Convention}. We use
the notation $t$ to
 denote a first-order type in
 $\mathbf{F}_{2}^{\mu}$ that represents the first-order term $t$. The term context $C$ containing
 one $\bullet$ can be represented as $\lambda x . C[x]$, a second-order type of kind $* \Rightarrow *$ 
 in $\mathbf{F}_{2}^{\mu}$. % Thus, $t$ could mean a 
% first-order term or a first-order type, depending on its surrounding textual environment. 
We use letters $F, G, D, S, Z$ to denote type
constants as well as function symbols. Note that for any first-order term $t$, it is always a well-kinded first-order type, since for any function
symbol $F^n$ in $t$, we can assign the kind $*^n \Rightarrow *$ for
$F$ and first-order term variable is of kind $*$. The following definition illustrates our use of this notation convention.
 
\begin{definition}[Leibniz representation]
  \label{leibniz}
Given a set of  rewrite rules $\mathcal{R}$,  we define the \textit{Leibniz representation} of $\mathcal{R}$ as $\mathbf{F}_{2}^{\mu}$-environments $\Gamma_{\mathcal{R}}, \Delta_{\mathcal{R}}$, as follows: 
\begin{itemize}
\item $\kappa : \forall p . \forall \underline{x} . p \ r \Rightarrow p \ l \in \Gamma_{\mathcal{R}}$ whenever $l \to r \in \mathcal{R}$, and where $\kappa$ is a fresh evidence constant and $\underline{x}$ are the free variables in $l$. 
\item $F :   *^n \Rightarrow * \in \Delta_{\mathcal{R}}$ if $F^n$ is a function symbol in $\mathcal{R}$. 
\end{itemize}
\end{definition}

Leibniz representation allows us to represent a first-order term rewriting system as a typing environment in $\mathbf{F}_2^\mu$, together with the typing rules, finite reductions can be represented by a typing judgement in $\mathbf{F}_2^\mu$. %% The following theorem shows that Leibniz representation can model the reduction of the first-order rewriting system. 

\begin{theorem}
\label{ade}
Let $\mathcal{R}$ be a set of rewrite rules. 
  \begin{enumerate}
  \item If $C[t] \to C[t']$ by $l \to r \in \mathcal{R}$, then
    $\Gamma_{\mathcal{R}} \vdash e : C[t'] \Rightarrow C[t]$ for some
     $e$.
   \item If $t_1 \to t_2 \to t_3$ is a reduction using $\mathcal{R}$, then $\Gamma_{\mathcal{R}} \vdash e : t_3 \Rightarrow t_1$ for some $e$. 
  \end{enumerate}
\end{theorem}
% \knote{quick check: we have `` $\in R$'' and "in $R$'' in items above. It is deliberate, because the latter item is not showing rewriting rules in $R$, but a reduction via rules of $R$?, (done. f.p.)}
\begin{proof}
  \begin{enumerate}
  \item By Definition \ref{leibniz}, we have $\kappa : \forall p. \forall
    \underline{x}. p \ r \Rightarrow p\ l \in
    \Gamma_{\mathcal{R}}$. We instantiate $p$ with 
    $\lambda y . C[y]$, by rule \textsc{Conv}, we get $\Gamma_{\mathcal{R}} \vdash \kappa \ (\lambda y . C[y]) : 
    \forall \underline{x}.C[r] \Rightarrow C[l]$.
    Since $\sigma l \equiv t, \sigma r \equiv
    t'$, let $\underline{t}$ be the codomain of $\sigma$, we have
    $\Gamma_{\mathcal{R}} \vdash \kappa\ \ (\lambda y
    . C[y])\ \underline{t} : C[t'] \Rightarrow C[t]$.
   \item By (1), we have $\Gamma_{\mathcal{R}} \vdash e_1 : t_2 \Rightarrow t_1$ and $\Gamma_{\mathcal{R}} \vdash e_2 : t_3 \Rightarrow t_2$, so $\Gamma_{\mathcal{R}} \vdash \lambda \alpha . e_1\ (e_2 \ \alpha) : t_3 \Rightarrow t_1$.
  \end{enumerate}
\end{proof}

\section{Hereditary Head Normalization and Faithfulness}
\label{meta}

In this section we define the \textit{hereditary head normalization} for an evidence  (Definition \ref{here}). The role of hereditary head normalization is
similar to \textit{productivity}, i.e. a hereditary head normalizing evidence can be associated with a computational tree (B\"ohm tree without bottom~\cite{barendregt1984lambda}). In $\mathbf{F}_2^\mu$, hereditary head normalization implies \textit{faithfulness}. Informally, an evidence is \textit{faithful} if we can recover a nonterminating reduction from it. 

To define hereditary head normalization, we first define an erasure that maps $\mathbf{F}_{2}^{\mu}$-evidence to pure lambda term with fixed point operator. 
\begin{definition}[Erasure]
\label{erasure}
  We define erasure $|\cdot |$ on evidence as follows. 

{\small
$ |\alpha| = \alpha
\quad 
 |\kappa| = \kappa
\quad 
 |\lambda \alpha . e| = \lambda \alpha . |e|
\quad 
|\mu \alpha . e| = \mu \alpha . |e|
\quad
|e\ e'| = |e|\ |e'|
\quad
\mgray{|\lambda x . e| = |e|}
\quad \mgray{|e\ T| = |e|} $
}
\end{definition}
We call the erased evidence $|e|$ \textit{Curry-style} evidence.
The following definition %of hereditary head normalization
follows the same formulation by Raffalli 
\cite{raffalli1994} and Tatsuta \cite{tatsuta2008}. 
\begin{definition}[Hereditary Head Normalization]
\label{here}
  Let $\Lambda$ be the set of Curry-style evidence. \noindent We say $e$ is \textit{hereditary head normalizing} (denoted by $e \in \mathrm{HHN}$) iff $|e| \in \mathrm{HN}_{n}$ for all $n \geq 0$.
We define $\mathrm{HN}_n$ as follows. 
  \begin{itemize}
  \item $e \in \mathrm{HN}_0 $ iff $e \in \Lambda$. 
  \item $e \in \mathrm{HN}_{n+1}$ iff $e \leadsto_{\beta\mu}^* \lambda \underline{\alpha}. e' \ e_1 \ ...\ e_m$, where $e'$ is a variable or a constant and $e_i \in \mathrm{HN}_{n}$ for all $i$.
  \end{itemize}

\end{definition}

We are going to show in Theorem \ref{faithfulness} that if $\Gamma_{\mathcal{R}} \vdash e : t$ in $\mathbf{F}_{2}^{\mu}$ and $e$ is hereditary head normalizing, then we can reconstruct a nonterminating reduction of $t$ by following the unfolding of $e$.
First we define the notion of \textit{trace}. The position of a trace is described as follows: 
Let $o$ denote the 
origin of a trace and $s\cdot m$ denote the next position after $m$. For a trace $\mathcal{T}$, we 
use $\mathcal{T}_m$ to refer to the node at position $m$ in the trace. 
The following formalization of \textit{evidence trace} is a degenerate case of B\"ohm tree (\cite{bohm1968alcune}, \cite[\S 10]{barendregt1984lambda}).  
\begin{definition}[Evidence Trace]
\label{trace}
Suppose $e \leadsto_{h\tau o}^* \kappa\ T_1 ... \ T_n \ e'$, with $T_1, ..., T_n$ in $\to_o$-normal form. The evidence trace of $e$, denoted by $[e]$, is defined as:
\begin{itemize}
\item $[e]_{o} = \kappa\ T_1 ... \ T_n$.
\item $[e]_{s \cdot m} = [e']_m$. 
\end{itemize}
\end{definition}

In the above definition, since $\kappa\ T_1 ... \ T_n \ e'$ is a head normal form, by the confluence of $\leadsto_{h\tau o}$ (Theorem \ref{confluence}), we know that $[e]$ is referring to at most one trace. When $e \not \leadsto_{h\tau o}^* \kappa\ T_1 ... \ T_n \ e'$, we say
$[e]$ is undefined. 
For an example of finite evidence trace, consider $e \equiv \kappa\ (\lambda y.y) \ (\kappa'\ (\lambda y.y))$, in this
case $[e]_o = \kappa\ (\lambda y.y), [e]_{s\cdot o} = \kappa' \ (\lambda y.y)$. For an example of an infinite evidence trace, consider
$e \equiv \mu \alpha . \kappa\ (\lambda y.y)\ \alpha$, we have $[e]_m = \kappa \ (\lambda y.y)$ for any position $m$. 

Intuitively, an evidence trace can be viewed as a sequence of instructions (in the form of evidence constants)  that we are going to follow in order to rewrite a term. The following definitions of \textit{action} and \textit{faithful action} on a first-order term reflects this intuition. 
Suppose $C[\sigma l, ... ,\sigma l ] \to^* C[\sigma r, ..., \sigma r]$ by $l \to r \in \mathcal{R}$.
We record the term context and the instantiation information along the reduction, i.e. $C[\sigma l, ..., \sigma l] \to_{(\kappa, C,\sigma)}^* C[\sigma r, ..., \sigma r]$. 
 
\begin{definition}[Action on First-Order Term]
\label{action}
Suppose $[e]_m = \kappa\ (\lambda x . C[x,...,x])\ t_1 ... \ t_n$ for some position $m$
and $\kappa : \forall p . \forall \underline{x} . p\ r \Rightarrow p\ l$. 
An \textit{action} of $[e]_m$ on the first-order term $t$ (denoted by $[e]_m(t)$) is defined as follows.
\begin{itemize}
\item $[e]_m(t) = t'$ if $t \to_{(\kappa, C,\sigma)}^* t'$, where $\sigma = [t_1/x_1,..., t_n/x_n]$. 
\item otherwise $[e]_m(t)$ is undefined. 
\end{itemize}
\end{definition}

Note that we write $t \to^* [e]_m(t)$ when $[e]_m(t)$ is defined. The following definition of 
\textit{faithful action} shows how one follows a potentially infinite evidence trace to
reduce a term. 

\begin{definition}[Faithful Action]
\label{faithful}
  The evidence trace $[e]$ acts on $t$ \textit{faithfully}, if we have a reduction sequence $t \to^* [e]_o(t) \to^* [e]_{s\cdot o}([e]_o(t)) \to^* [e]_{s\cdot s \cdot o}([e]_{s\cdot o}([e]_o(t))) \to^* ... \to^* [e]_m(...[e]_o(t)...)$ for any position $m$. 
\end{definition}

%\knote{I polished the style in the below example, a bit. Check me?}

%\knote{I got confused at first but the below example, -- it looked like the Section 2 example is repeated once again.
 %%  I wonder if it makes sense to just come back to that running example here, and explain properly how the action of trace on term you showed on Section 2 agrees with the above formal definitions.
%%   This I think will be very helpful.
%%   Otherwise, the reader needs to digest a new, but very similar, set of formulas and contexts... 
%% } make sense.
\begin{example}
  To illustrate the intuition behind Definitions \ref{trace}, \ref{action}, \ref{faithful}, let us consider the one rule rewriting system: $F\ x \to G\ (F \ (G\ x))$ in Example \ref{gg}. The Leibniz representation is $\Delta = F : * \Rightarrow *, G : * \Rightarrow *, \Gamma =  \kappa : \forall p . \forall x . p \ (G\ (F\ \ (G \ x))) \Rightarrow p \ (F\ x)$. Recall that we had the following judgement. 

  \begin{center}
  {\footnotesize    
    
    (1)  $\Gamma \vdash e' \equiv \mu \alpha . \lambda p . \lambda x . \kappa\ p\ x\
  (\alpha\ (\lambda y . p\ (G\ y))\ (G\ x)) : \forall p. \forall x . p
    \ (F \ x)$

    \
    
    (2)  $\Gamma \vdash e' \ (\lambda y. y)\ x  : F \ x$
  }

%%  (1) $\Gamma, x : *, \alpha : \forall p . p\ (F\ x) \vdash \lambda p.
%%     \kappa \ p\ x \ (\alpha \ (\lambda y . p \ (G\ y))) : \forall p
%%     . p\ (F\ x)$

%% \

%%   (2)  $\Gamma, x : * \vdash e' \equiv \mu \alpha . \lambda p.
%%     \kappa \ p\ x \ (\alpha \ (\lambda y . p \ (G\ y))) : \forall p
%%     . p\ (F\ x)$

%% \

%% (3) $\Gamma, x : * \vdash e \equiv (\mu \alpha . \lambda p.  \kappa \ p \ x \ (\alpha \ (\lambda y . p \ (G\ y)))) \ (\lambda y . y) : F\ x$
%% }
  \end{center}
  We observed the following unfolding of  $e' \ (\lambda y. y)\ x $
(below $\mathcal{C} \equiv \kappa \ (\lambda y . y) \ x \ (\kappa \ (\lambda y . G \ y)\ (G\ x)\ \bullet)$):
  % shows the nontermination of $F\ x$ and we have the following evidence reduction: 

\begin{center}
{\footnotesize  $e' \ (\lambda y. y)\ x \leadsto_{\beta\mu\tau o}^* \mgray{\kappa \ (\lambda y . y) \ x}\ (e'\ (\lambda y . G\ y)\ (G\ x))
  \leadsto_{\beta\mu\tau o}^* \kappa \ (\lambda y . y) \ x \ (\mgray{\kappa \ (\lambda y . G \ y)\ (G\ x)}\
  (e'\ (\lambda y . G \ (G\ y))) \ (G\ (G\ x))) \leadsto_{\beta\mu\tau o}^* \mathcal{C}[(\mgray{\kappa \ (\lambda y . G \ (G \
  y)) \ (G \ (G \ x))}\ (e'\ (\lambda y . G \ (G \ (G\ y))))\ (G\ (G\ (G \ x))))] \leadsto_{\beta\mu\tau o}^* ... $}
\end{center}
 
\noindent It gives rise to the following evidence trace: $[e]_o = \kappa \ (\lambda y . y)\ x$, $[e]_{s \cdot o} = \kappa \ (\lambda y . G \ y)\ (G\ x)$,
$[e]_{s \cdot s \cdot o} = \kappa \ (\lambda y . G\ (G \ y))\ (G \ (G \ x))$, etc.
Moreover $[e]$ acts faithfully on $F\ x$  (by Theorem \ref{faithfulness}). For example, we observe that $F\ x \to [e]_o(F\ x) \to [e]_{s \cdot o}([e]_{o}(F\ x)) \to [e]_{s \cdot s \cdot o}([e]_{s \cdot o}([e]_{o}(F\ x)))$, which is the following reduction trace. 
\begin{center}
  $F\ x \to_{(\kappa, \bullet, [x/x])} G\ (F\ (G \ x)) \to_{(\kappa, G\ \bullet, [(G \ x)/x])} G\ (G\ (F\  (G \ (G\ x)))) \to_{(\kappa, G \ (G\ \bullet), [(G\ (G\ x))/x])} G \ (G\ (G\ (F\ (G\ (G\ (G\ x))))))$
\end{center}

\end{example}

% Now we are going to show well-typed evidence in $\mathbf{F}_{2}^{\mu}$ acts on the corresponding first-order term faithfully. Theorem \ref{faithfulness} essentially is a consequence of the head normalization of the evidence in $\mathbf{F}_{2}^{\mu}$. 
%% Let $t$ be a first-order term, we write $\Gamma_t$ to denote an $\mathbf{F}_{2}^{\mu}$ environment such that for any free variable $x \in \mathrm{FV}(t)$, we have $x : * \in \Gamma_t$. 
%\fp{need to rephrase}
\begin{lemma}
\label{inversion}
Suppose $\Gamma_{\mathcal{R}} \vdash e : t$ for some first-order term $t$ and $e$ is head normalizing. We have $e \leadsto_{h\tau o}^* \kappa \ (\lambda x . C[x,..., x]) \ t_1 ...\ t_n\  e'$ for some $\kappa :  \forall p . \forall \underline{x}. p \ r \Rightarrow p\ l \in \Gamma_{\mathcal{R}}$.
Furthermore, we have $\Gamma_{\mathcal{R}} \vdash e' :  C[\sigma r, ..., \sigma r]$ and $C[\sigma l, ..., \sigma l] = t$, where $\mathrm{codom}(\sigma) = \{t_1,..., t_n\}$ and $\mathrm{dom}(\sigma) = \mathrm{FV}(l)$. 
\end{lemma}

\begin{theorem}[Faithfulness of Corecursive Evidence]
\label{faithfulness}
Suppose $\Gamma_{\mathcal{R}} \vdash e : t$ in $\mathbf{F}_2^\mu$ and $e \in \mathrm{HHN}$. We have $t \to^* [e]_o(t) \to^* ...\to^* [e]_m(...[e]_o(t)...)$ for any position $m$, i.e. $e$ acts faithfully on $t$. 
\end{theorem}
\begin{proof}
  By Lemma \ref{inversion}, we know that $e \leadsto_{h\tau o}^* \kappa \ (\lambda x . C[x,..., x]) \ t_1 ...\ t_n\  e'$ for some $\kappa :  \forall p . \forall \underline{x}. p \ r \Rightarrow p\ l$, $\Gamma_{\mathcal{R}} \vdash e' :  C[\sigma r, ..., \sigma r]$, $C[\sigma l, ..., \sigma l] = t$, where $\mathrm{codom}(\sigma) = \{t_1,...,t_n\}$ and $\mathrm{dom}(\sigma) = \mathrm{FV}(l)$.
  Thus $t = C[\sigma l,..., \sigma l] \to_{(\kappa, C, \sigma)}^* \ C[\sigma r,..., \sigma r]$.  We prove the theorem by induction on $m$. 
  \begin{itemize}
  \item $m = o$. We have $[e]_o = \kappa \ (\lambda x . C[x,..., x]) \ t_1 ...\ t_n$, since $t \to_{(\kappa, C, \sigma)}^* \ C[\sigma r,..., \sigma r]$, so $t \to^* [e]_o(t)$.
   \item $m = s \cdot m'$. We need to show $t \to^* [e]_o(t) \to^* ... \to^* [e]_{s\cdot m'}(... [e]_o(t) ...)$. 
Since $\Gamma_{\mathcal{R}} \vdash e' :  C[\sigma r,..., \sigma r]$ and $e' \in \mathrm{HHN}$, by IH, we have $C[\sigma r,..., \sigma r] \to^* [e']_o(C[\sigma r,..., \sigma r]) \to^* ... \to^* [e']_{m'}(... [e']_o(C[\sigma r]) ...)$. Thus $t \to^* [e]_o(t) = C[\sigma r,..., \sigma r] \to^* [e']_o([e]_o(t)) \to^* ... \to^* [e']_{m'}(... [e']_o([e]_o(t)) ...)$. Since $[e']_{a} = [e]_{s\cdot a}$ for any position $a$, we have $t \to^* [e]_o(t) \to^* [e]_{s \cdot o}([e]_o(t))\to^* ... \to^* [e]_{s \cdot m'}(... [e]_{s \cdot o}([e]_o(t)) ...)$.  
  \end{itemize}
\end{proof}

\begin{comment}
The above faithfulness theorem does not hold for $\mathbf{F}_{\omega}^\mu$ ($\mathbf{F}_{\omega}$ with fixpoint typing rule). Consider the terminating rewrite system
\footnote{We thank FSCD 2016 Reviewer 1 for this example.} 
$\kappa : A\ x \to B\ x$, and its Leibniz
representation  $\Gamma = A : * \Rightarrow *, B : * \Rightarrow *, \kappa : \forall p . \forall x . p\ (B\ x) \Rightarrow p\ (A \ x)$. We have the following judgements in $\mathbf{F}_{\omega}^\mu$: 

\begin{center}
  {\footnotesize
  $\Gamma, x : * \vdash \kappa\ (\lambda y . y \Rightarrow y) : (B\ x \Rightarrow B\ x) \Rightarrow (A\ x \Rightarrow A\ x)$

  \
  
  $\Gamma, x : * , \beta : A\ x \vdash \kappa\ (\lambda y . y \Rightarrow y) \ (\lambda \alpha . \alpha) \ \beta :  A\ x$

  \
  
  $\Gamma, x : * \vdash \mu \beta. \kappa\ (\lambda y . y \Rightarrow y) \ (\lambda \alpha . \alpha) \ \beta :  A\ x$
  }
\end{center}
\noindent Note that $|\mu \beta. \kappa\ (\lambda y . y \Rightarrow y) \ (\lambda \alpha . \alpha) \ \beta| = \mu \beta . \kappa\ (\lambda \alpha. \alpha) \ \beta \in \mathrm{HHN}$, but it certainly does not represent a nonterminating reduction
beginning with $A\ x$. We can see this kind of construction is not possible in $\mathbf{F}_2^\mu$
as $\lambda y. y \Rightarrow y$ is not kindable and it does not represent a first-order term context.%\footnote{This example is due to a reviewer at FSCD 2016}. 
\end{comment}
Now we are going to show the hereditary head normalization for $\mathbf{F}_{2}^{\mu}$ is decidable
by mapping a typable evidence in $\mathbf{F}_{2}^{\mu}$ to a typable evidence in $\lambda$-Y caculus (simply typed lambda calculus with fixpoint typing rule \cite{Statman04})\footnote{Please see Appendix \ref{lambda-y} for full details.}. 

\begin{definition}
  We define a function $\theta$ that maps $\mathbf{F}_{2}^\mu$ types to $\lambda$-Y types.
  \
  
{\small
\noindent  $\theta(x | F) = B$
  \quad
  $\theta(\lambda x. T) = \theta(T)$
\quad
  $\theta(T\ T') = \theta(T)$ \quad
  $\theta(T \Rightarrow T') = \theta(T) \Rightarrow \theta(T')$
\quad
  $\theta(\forall x.T) = \theta(T)$}
\end{definition}
We write $\theta(\Gamma)$ to mean applying the function $\theta$ to all the types in $\Gamma$.
Type $B$ is the based type in $\lambda$-Y. 

\begin{theorem}
  \label{2toy}
  If $\Gamma \vdash e : T$ and $\Delta \vdash T : *|o $ in $\mathbf{F}_{2}^{\mu}$, then
  $\theta(\Gamma)\vdash |e| : \theta(T)$ in $\lambda$-Y.
\end{theorem}

Theorem \ref{2toy} implies that the hereditary head normalization for $\mathbf{F}_2^\mu$ is
decidable, since it is well-known that hereditary head normalization for $\lambda$-Y is decidable (\cite{broadbent2010recursion}, \cite{serre2015playing}, \cite{grellois2016semantics}).

 \section{Type Checking $\mathbf{F}_{2}^{\mu}$ Based on Resolution with Second-order Matching}
 \label{typecheck}
%We use this section to discuss the challenges of functional certification of nonlooping nontermination reductions, using the following example:
Modeling first-order term contexts is one of the reasons we use
second-order types. Quantification over second-order type variables also enables us to represent some \emph{nonlooping} nonterminations in $\mathbf{F}_{2}^{\mu}$. 

\begin{example}
  \label{tom}

Consider the following rewrite rules \cite{fu2016}. 
\begin{center}
  {\footnotesize
  $ D\ (S\ x)\ y \to_a D\ x\ (S\ y)$

  $ D\ Z\ y \to_b D\ (S\ y)\ Z$}
\end{center}
\noindent The term $D\ Z\ Z$ will give rise to the following nonlooping nonterminating 
reduction, where no cycle or loop can be observed: 

\begin{center}
{\footnotesize  $D \ Z \ Z \to_b D \ (S\ Z) \ Z \to_a D \ Z \ (S\ Z) \to_b D \ (S \ (S\ Z)) \ Z \to_a D \ (S\ Z) \ (S \ Z) 
  \to_a D \ Z\ (S \ (S\ Z))  \to_b  D \ (S \ (S \ (S\ Z)))\ Z \to_a D \ (S \ (S\ Z))\ (S\ Z) \to_a 
D \ (S\ Z)\ (S\ (S\ Z)) \to_a D \ Z\ (S\ (S \ (S\ Z))) \to ...$
}
\end{center}
The rule sequence for this reduction exhibits the pattern: ``$ba, baa, baaa,...$'', which can be represented
by the corecursive function $f\ \alpha\ \beta = (\beta \cdot \alpha) \ (f\ \alpha\ (\beta \cdot \alpha))$(here $\cdot$ denotes functional composition), as $f\ a\ b$ would give rise to the following reduction (we omit the compositional symbols): 

\begin{center}
  \noindent {$f \ a \ b \leadsto (b a) (f\ a\ (b a))
    \leadsto (b a b a a) (f \ a\ (b a a))\leadsto (b a b a a b a a a) (f \
    a\ (b a a a)) $}
\end{center}

%% Coming back to the discussion of guardedness, we note that
%% the corecursive call to $f$ is under the input variable $\alpha$ (rather than a constant), yet it is considered guarded by our definitions.
\noindent Let the Leibniz representation of the rewriting system be as follows: 

\begin{center}
  {\footnotesize $\Delta = D : *^2 \Rightarrow *, Z : * , S
    : * \Rightarrow *$ \\
    $\Gamma = \kappa_a : \forall p . \forall x . \forall y
    . p\ (D\ x\ (S\ y)) \Rightarrow p \ (D\ (S\ x)\ y), \kappa_b :
    \forall p . \forall y. p\ (D\ (S\ y)\ Z) \Rightarrow p \ (D\ Z\
    y)$}
\end{center}

\noindent We would like to provide a type annotation for $f$ such that
 $\Gamma \vdash f\ \kappa_a \ \kappa_b : D \ Z \ Z$.
 % to come up with the right type signature for $f$,
But it is not obvious as we cannot type check $f\ \kappa_a \ \kappa_b$ with $D \ Z\ Z$ using any first-order type checking algorithm (e.g. the one in Haskell).
We will show how to type check $f$ using the type checking algorithm we introduce in this section.
\end{example}

%% Since all  type variables in $\mathbf{F}_{2}^{\mu}$ are at most second-order, this provides opportunity for us to implement
%% a type checking algorithm for $\mathbf{F}_{2}^{\mu}$ using resolution with second-order matching.
By \textit{type checking}, we
mean the following problem: given an environment $\Gamma$, a Curry-style evidence $e$ and a type $T$, construct a fully annotated evidence $e'$ such that $\Gamma \vdash e' : T$ and $|e'| = e$. We use the terminology 
\textit{proof checking} to mean the following: given an environment $\Gamma$, a fully annotated evidence $e$ and a type $T$,
check if $\Gamma \vdash e : T$. 
 The type checking problem for Curry-style System $\mathbf{F}$ and $\mathbf{F}_\omega$ are well-known to be undecidable (\cite{wells1999}, \cite{urzyczyn1997}).
 The type system $\mathbf{F}_2^\mu$ appears to be a much weaker system compared
 to System $\mathbf{F}$ and $\mathbf{F}_\omega$ (HHN is decidable in $\mathbf{F}_2^\mu$), we will show a type checking algorithm for $\mathbf{F}_2^\mu$ inspired by SLD-resolution~\cite{nilsson1990logic}. We will work on types 
 that are kindable by our decidable kind system (Definition \ref{kindsystem}).
 Moreover, we will consider the following reformulation of type
 $T$ from Definition~\ref{syntax}:
 $$T ::=  A  \ | \  \forall \underline{x}. T \Rightarrow ... \Rightarrow T \Rightarrow A$$
 Here $A$ is
 of kind $*$. We use $T_1,..., T_n \Rightarrow A$ as a shorthand for $T_1 \Rightarrow ... \Rightarrow T_n \Rightarrow A$ and we call $A$ the \textit{head} of $T_1,..., T_n \Rightarrow A$.
 These types are a generalized version of Horn formulas, called \textit{hereditary Harrop formula} in the literature \cite{miller1991uniform}. %% ), they can express all the universal and implicative intuitionistic formulas. 

 In this section we use $A, B$ to denote a type of kind $*$, and we use $a, b$ to denote a type
 variable or a type constant. The following definition of second-order matching 
 follows Dowek's treatment \cite{dowek2001} of Huet's algorithm \cite{gerard1976}. 
 \begin{definition}[Second-order Matching]
   \label{som}
Let $E$ be a set of second-order matching problems $\{A_1 \mapsto B_1,..., A_n \mapsto B_n\}$. The following 
rules (intended to apply top-down) show how to transform $E$. 

\

{\footnotesize

  \begin{tabular}{lll}
 \infer{\bot}{\{F \ A_1\ ...\ A_n \mapsto G \ B_1 \ ... \ B_m, E \}}

    & &
    
    \infer{\{A_1 \mapsto B_1,..., A_n \mapsto B_n, E\} }
          {\{a \ A_1\ ...\ A_n \mapsto a \ B_1 \ ... \ B_n, E\} }
          
          \\  \\
          
          \infer[\textsc{Proj}]{[(\lambda x_1. .... \lambda x_n . x_i/y]E' }{E' \equiv \{y \ A_1\ ...\ A_n \mapsto a \ B_1 \ ... \ B_m, E\}}

  & &

            \infer[\textsc{Imi}]{[(\lambda x_1. .... \lambda x_n . a \ (y_1 \ \underline{x}) ...\ (y_m \ \underline{x}))/y]E' }{E' \equiv \{y \ A_1\ ...\ A_n \mapsto a \ B_1 \ ... \ B_m, E\}}
  \end{tabular}
     
}          
 \end{definition}

 Note that $\bot$ denotes a failure in matching. In the \textsc{Imi} rule, the variables $y_1,..., y_m$ are fresh type variables. 
 The \textsc{Proj} and \textsc{Imi} rules introduce
  nondeterminism, so there may be multiple matchers for a matching problem $A \mapsto B$. We
 write $A \mapsto_\sigma B$ to mean there is a derivation from $A \mapsto B$ to $\emptyset$ 
 using rules in the above definition with a second-order matcher $\sigma$. %% , as we want to know whether $A$ is matchable
 %% to $B$ as well as the matcher $\sigma$.
 %% If $A \mapsto_\sigma B$, then $\sigma A \to^*_o B$. Note that
The second-order matching is decidable (all derivations are finite using Definition \ref{som}) and all the resulted matchers are finite, but second-order unification is not decidable \cite{goldfarb1981}.

 The standard second-order matching algorithm usually generates
 many vacuous substitutions, we can
 exclude them by kinding, as we work with kindable types.
For example, when we match
$d \ Z \ Z$ against $D \ Z \ (S\ Z)$, the second-order matching algorithm would
 generate matchers such as $[\lambda x. \lambda y. D\ Z \ (S\ Z)/d]$ and $[\lambda x. \lambda y. D\ y \ (S\ y)/d]$, which are not kindable. %% Note that 
 %% the second-order matching is still sound with this restriction (i.e. if $A \mapsto_\sigma B$, then $\sigma A \to_o^* B$), but it will be incomplete.  

 %\knote{Perhaps somewhere here is a good place to add a bit of discussion of : 1) relation of this formalism to Horn clauses and resolution;
 %  b) relation of this to resolution in Haskell c) that all formulas we consider are representable as Horn clauses, -- hence resolution is an adequate method.
%Otherwise, resolution is in the centre of attention for 2 sections, but it is not motivated as well as other concepts introduced in other sections
% }

 Let $T = \forall x_1. ...\forall x_m. T_1, ... , T_n \Rightarrow A$,
 the set of variables $\{x_i \ | \ x_i \notin \mathrm{FV}(A), 1 \leq i \leq m\} \cup \mathrm{FV}(T)$ are called \emph{existential variables}. In this section,  we work with types that do not have any existential variables, we will show how to handle existential variables in the next section.
We use $\Phi$ to
denote a set of tuples of the form $(\Gamma, e, T)$. %%  where $T$ is the current goal for the
%% resolution and  $e$ is a Curry-style evidence that
%% intuitively can be understood  as a list of instructions for the resolution algorithm.  
We define \textit{resolution by second-order matching} as a transition system from 
$\Phi$ to $\Phi'$ as follows:

\begin{definition}[Resolution by Second-order Matching (RSM)]
\label{rsm}
  \fbox{$\Phi \longrightarrow \Phi'$}

  \begin{enumerate}
  \item  $\{(\Gamma, (\kappa | \alpha) \ e_1 \ ... \ e_n, A), \Phi \} \longrightarrow_a \{(\Gamma,e_1, \sigma T_1), ..., (\Gamma,e_n, \sigma T_n), \Phi\}$
    if $\kappa | \alpha : \forall \underline{x}. T_1,..., T_n \Rightarrow B \in \Gamma$ with $B \mapsto_\sigma A$.
  \item $\{(\Gamma, \lambda \alpha_1. ... \lambda \alpha_n. e, T_1, ..., T_n \Rightarrow A), \Phi\} \longrightarrow_i \{([\Gamma,\alpha_1 : T_1, ..., \alpha_n : T_n], e, A), \Phi \}$.

      \item $\{(\Gamma, e, \forall x_1 ... \forall x_n. T), \Phi\} \longrightarrow_{\forall} \{(\Gamma, e, T), \Phi \}$.

  \item $\{(\Gamma, \mu \alpha. e, T), \Phi\} \longrightarrow_c \{([\Gamma, \alpha : T], e, T), \Phi \}$.

%  \item $(\Gamma, a \ e_1 ... \ e_n, [A ,\Delta]) \longrightarrow_{a}$.
  %% \item $\{(\Gamma, a, [B]), \Phi\} \longrightarrow_a \{\Phi\} $ if $a \equiv \kappa | \alpha$ and $a : B \in \Gamma$. 
  
  \end{enumerate}
\end{definition}

As before,  $\kappa | \alpha$ means ``$\kappa$ or $\alpha$''.
The rule (1) allow the the size of $\{e_1,..., e_n\}$ to be zero. %% The rule (2) will introduce eigenvariables $x_1,..., x_m$ (which behave as constants) for $T_1,..., T_n, A$, and
We require the sizes of $\{\alpha_1,..., \alpha_n\}$ and $\{x_1,..., x_n\}$ both to be nonzero for rules (2) and (3). Rule (3) also introduces fresh \textit{eigenvariables} $\{x_1,..., x_n\}$ for $T$, they behave the same as constants during RSM. In rule (1), when perform matching $B \mapsto_\sigma A$, we rename the bound variables $\underline{x}$ in $T_1,..., T_n, B$ to fresh variables. 
The $T$ in the tuple $(\Gamma, e, T)$ intuitively corresponds to the current goal for the
resolution and $e$ is a Curry-style evidence that
 can be understood as a list of instructions for the resolution algorithm.  
The resolution is defined by case analysis on the Curry-style evidence and the current goal $T$ and it is terminating. If it terminates with the empty
set, then we say the resolution succeeds, otherwise it fails. The following theorem shows that if the resolution succeeds, then the type checking succeeds, i.e. we can obtain the corresponding fully annotated evidence.

\begin{theorem}[Soundness of RSM]
\label{thm:sound}
If $\{(\Gamma, e, T)\} \longrightarrow^* \emptyset$, then there exists an evidence $e'$ such that $\Gamma \vdash e' : T$ in $\mathbf{F}_{2}^{\mu}$ and $|e'| = e$.
  \end{theorem}

The proof of Theorem \ref{thm:sound} gives us an algorithm to compute the annotated 
evidence $e'$. This algorithm is implemented in our prototype. 

\begin{example}\label{ex:tom2}
  Continuing the Example~\ref{tom}, let us illustrate how to use RSM to type check the function $f$. Consider the long form of $f$, namely, $f   =  \mu f. \lambda \alpha . \lambda \beta. \beta (\alpha \ (f\ (\lambda \alpha'. \alpha \ \alpha')\ (\lambda \alpha'. \beta( \alpha \ \alpha')))$ and the Leibniz representation:
  
  \begin{center}
    {\footnotesize $\Delta = D : *^2 \Rightarrow *, Z : * , S
      : * \Rightarrow *$

      $\Gamma = \kappa_a : \forall p . \forall x . \forall y
      . p\ (D\ x\ (S\ y)) \Rightarrow p \ (D\ (S\ x)\ y), \kappa_b :
      \forall p . \forall y. p\ (D\ (S\ y)\ Z) \Rightarrow p \ (D\ Z\
      y)$.}
  \end{center}
  As we want $\Gamma \vdash f \ \kappa_a \ \kappa_b : D \ Z \ Z$, the most intuitive type that we can assign to $f$ is the following. 

\begin{center} {\small $ T\equiv (\forall p . \forall x . \forall y
    . p \ (D \ x \ (S\ y)) \Rightarrow p \ (D\ (S\ x) \ y))
    \Rightarrow (\forall p . \forall y. p\ (D\ (S\ y)\ Z) \Rightarrow p \
    (D\ Z\ y)) \Rightarrow D\ Z\ Z$}
\end{center}
But $f$ can not be type checked with $T$ by RSM.  
The solution is abstracting $D$ to a second-order variable $d$ and assigning the following type to $f$: 

\begin{center} {\small $T' \equiv \forall d. \underbrace{(\forall p . \forall x . \forall y
    . p \ (d \ x \ (S\ y)) \Rightarrow p \ (d\ (S\ x) \ y))
    \Rightarrow (\forall p . \forall y. p\ (d\ (S\ y)\ Z) \Rightarrow p \
    (d\ Z\ y)) \Rightarrow d\ Z\ Z}_{T''}$}
\end{center}
This change yields the following successful RSM resolution trace. 

\begin{center}
  {
    $\{(\Gamma,\mu f. \lambda \alpha . \lambda \beta. \beta (\alpha \ (f\ (\lambda \alpha'. \alpha \ \alpha')\ (\lambda \alpha'. \beta( \alpha \ \alpha')))), T')\} \longrightarrow_c$

    $\{([\Gamma, f : T'], \lambda \alpha . \lambda \beta. \beta (\alpha \ (f\ (\lambda \alpha'. \alpha \ \alpha')\ (\lambda \alpha'. \beta( \alpha \ \alpha')))), T')\} \longrightarrow_{\forall}$
    
    $\{([\Gamma, f : T'], \lambda \alpha . \lambda \beta. \beta (\alpha \ (f\ (\lambda \alpha'. \alpha \ \alpha')\ (\lambda \alpha'. \beta( \alpha \ \alpha')))), [d_1/d]T'')\} \longrightarrow_{i}
    \{(\Gamma'', \beta (\alpha \ (f\ (\lambda \alpha'. \alpha \ \alpha')\ (\lambda \alpha'. \beta( \alpha \ \alpha')))), d_1\ Z \ Z)\} \longrightarrow_a \{(\Gamma'', \alpha \ (f\ (\lambda \alpha'. \alpha \ \alpha')\ (\lambda \alpha'. \beta( \alpha \ \alpha'))), d_1\ (S\ Z) \ Z)\} \longrightarrow_a
    \{(\Gamma'', f\ (\lambda \alpha'. \alpha \ \alpha')\ (\lambda \alpha'. \beta( \alpha \ \alpha')), d_1\ Z \ (S \ Z))\} \mgray{\longrightarrow_{a}}$

    $\{(\Gamma'', \lambda \alpha'. \alpha \ \alpha', \forall p . \forall x . \forall y
    . p \ (d_1 \ x \ (S\ (S \ y))) \Rightarrow p \ (d_1\ (S\ x) \ (S\ y))), \Phi_1 \equiv (\Gamma'', \lambda \alpha'. \beta ( \alpha \ \alpha'), \forall p . \forall y. p\ (d_1\ (S\ y)\ (S \ Z)) \Rightarrow p \ (d_1\ Z\ (S\ y)))\} \longrightarrow_{\forall}$

        $\{(\Gamma'', \lambda \alpha'. \alpha \ \alpha', 
     p_1 \ (d_1 \ x_1 \ (S\ (S \ y_1))) \Rightarrow p_1 \ (d_1\ (S\ x_1) \ (S\ y_1))), \Phi_1\}
     \longrightarrow_i$
     
    $\{([\Gamma'', \alpha' : p_1 \ (d_1 \ x_1 \ (S\ (S\ y_1)))], \alpha \ \alpha', 
     p_1 \ (d_1\ (S\ x_1) \ (S\ y_1))), \Phi_1\}
     \longrightarrow_a$ 

     $\{([\Gamma'', \alpha' : p_1 \ (d_1 \ x_1 \ (S\ (S\ y_1)))], \alpha', 
     p_1 \ (d_1\ x_1 \ (S\ (S\ y_1)))), \Phi_1\}
     \longrightarrow_a
     \{(\Gamma'', \lambda \alpha'. \beta ( \alpha \ \alpha'), \forall p . \forall y. p\ (d_1\ (S\ y)\ (S \ Z)) \Rightarrow p \ (d_1\ Z\ (S\ y)))\}
     \longrightarrow_{\forall}$

     $\{(\Gamma'', \lambda \alpha'. \beta ( \alpha \ \alpha'), p_2\ (d_1\ (S\ y_2)\ (S \ Z)) \Rightarrow p_2 \ (d_1\ Z\ (S\ y_2)))\} \longrightarrow_{i}$

      $\{([\Gamma'', \alpha' : p_2\ (d_1\ (S\ y_2)\ (S \ Z))], \beta ( \alpha \ \alpha'),  p_2 \ (d_1\ Z\ (S\ y_2)))\} \longrightarrow_a$ 

      %% $\{([\Gamma'', p_2 : * \Rightarrow *, y_2 : *,\alpha' : p_2\ (d_1\ (S\ y_2)\ (S \ Z))], \beta ( \alpha \ \alpha'),  p_2 \ (d_1\ Z\ (S\ y_2)))\} \longrightarrow_a$ 

      $\{([\Gamma'', \alpha' : p_2\ (d_1\ (S\ y_2)\ (S \ Z))], \alpha \ \alpha',  p_2 \ (d_1\ (S\ (S\ y_2)) \ Z))\} \longrightarrow_a$ 

      $\{([\Gamma'', \alpha' : p_2\ (d_1\ (S\ y_2)\ (S \ Z))], \alpha',  p_2 \ (d_1\ (S\ y_2) \ (S \ Z)))\} \longrightarrow_a \emptyset
    $}
\end{center}
Note that $\Gamma'' = \Gamma, f : T', \alpha : \forall p . \forall x . \forall y . p \ (d_1 \ x \ (S\ y)) \Rightarrow p \ (d_1\ (S\ x) \ y), \beta : \forall p . \forall y. p\ (d_1\ (S\ y)\ Z) \Rightarrow p \ (d_1\ Z\ y)$. At the third $\longrightarrow_a$-step, by second-order matching, we instantiate the $d$ in the type of $f$ to $\lambda x.\lambda y. d_1 \ x \ (S\ y)$. Now that $f$ is typable with $T'$, we have $\Gamma \vdash f\ D\ \kappa_a \ \kappa_b : D \ Z\ Z$. Since the rewriting system is non-overlapping and $f$ is hereditary head normalizing, by
Theorem \ref{faithfulness} we know $f\ D\ \kappa_a \ \kappa_b$ represents the nonterminating reduction of $D\ Z\ Z$.  
%% The ability of quantifying over second-order variables in the type $T'$ allows us to obtain the successful RSM trace. 
\end{example}
%% \knote{I took the below para out of example environment, as it seems to be a general discussion}

Representing nonterminations in general follows the same method as the above example:
one first writes down a corecursive function that 
represents the rule sequence in a nonterminating reduction, and then provides the
proper type signature for such function. Once the function is type checked,  
a finite representation can be obtained. We illustrate how the prototype works for
this example and some other challenging examples in the Appendix \ref{app:expaper}, \ref{app:rewritingex}.
%file, we demonstrate
%this example using our type checker.
%\knote{should (productive) be (guarded) in the para above?} 
%I think productive seems more accurate as sometimes guardness can be too restricted.

\section{RSM Algorithm with Existential Variables}
\label{heuristic}

%we called the free variable in $T_1,..., T_n$ that does not occur in $A$ \textit{existential variable}. %With the appearance of
The RSM algorithm in Definition \ref{rsm} fails to type check some judgements in presence of existential variables.
In this section, we extend RSM to cope with existential variables. 
As a result, the nontermination reduction in the Example \ref{fib} can also be type checked.
  %% we show that quantifying over second-order existential variables gives us the flexibility to represent a wider range of nonterminating reductions. %To see how this is the case,  
%% Consider the following two rewrite rules on strings. It is
%% taken from the literature studying L-systems\footnote{These rules are from the L-system tutorial in wikipedia: \url{https://en.wikipedia.org/wiki/L-system}} -- biologically-motivated parallel string rewriting systems:

%% \begin{center}
%%   $A \to_a AB$

%% $B \to_b A$
%% \end{center}
%% The following parallel reduction (where all redexes are reduced at one $\Longrightarrow$-step) can be used
%% %Here we are interested in using corecursive evidence
%% to  model the algae 
%% growth pattern.

%% \begin{center}
%%   ${{B}} \Longrightarrow {{A}} \Longrightarrow {{A}B} \Longrightarrow {{A}BA} \Longrightarrow {{A}BAAB} \Longrightarrow
%%   {{A}BAABABA} \Longrightarrow {{A}BAABABAABAAB} \Longrightarrow ...$
%% \end{center}
We consider the following sequential reduction that simulates the parallel reduction sequence in
the Example \ref{fib}.
%In particular,  the following sequential reduction  simulates the parallel reduction strategy of L-system.
%\knote{L-System is undefined and not explained in any way. So, I am not sure what it is}
 At each reduction step, we underline the chosen redex.
%\knote{Did we introduce the term ``redex'' before?}

 \begin{center}
   {\footnotesize
  $\mgray{\underline{A}} \to_a
  \mgray{\underline{A}B} \to_a AB\underline{B} \to_b
  \mgray{\underline{A}BA} \to_a AB\underline{B}A \to_b
  ABA\underline{A} \to_a \mgray{\underline{A}BAAB} \to_a
  AB\underline{B}AAB \to_b ABA\underline{A}AB \to_a
  ABAAB\underline{A}B \to_a ABAABAB\underline{B} \to_b
  \mgray{\underline{A}BAABABA} \to_a AB\underline{B}AABABA \to_b
  ABA\underline{A}ABABA \to_a ABAAB\underline{A}BABA \to_a
  ABAABAB\underline{B}ABA \to_b ABAABABA\underline{A}BA \to_a
  ABAABABAAB\underline{B}A \to_b ABAABABAABA\underline{A} \to_a
  \mgray{\underline{A}BAABABAABAAB} \to ...$}
\end{center}
Observe that the length of the gray strings grows according to the Fibonacci sequence, and each gray string is a result of concatenation of the previous two. 

The rule sequence in the above reduction is ``$a, ab, aba, abaab, abaababa$'' (each word in the rule sequence is a concatenation of the previous two).
We can use the corecursive function $f \alpha\ \beta = \alpha \ (f \ (\alpha \cdot \beta) \ \alpha)$ to generate such sequences. 

\begin{center}
  $f\ a \ b \leadsto a (f \ (ab)\ a) \leadsto 
  (a ab) (f\ (aba)\ (a b)) \leadsto (a ab aba) (f\ (abaab)\ (aba))$
\end{center}

We can use a standard method~\cite{bezem2003term} to represent string rewriting systems
as first-order term rewriting systems. In this case, the corresponding rules would be $A\ x \to_a A\ (B\ x)$ and $B\ x \to_b A\ x$. The reduction would begin with $A\ x$.
The Leibniz representation for this rewrite system is the following:  

\begin{center}
  {\footnotesize
  $\Delta = A : * \Rightarrow *, B : * \Rightarrow *$ \\
  
  $\Gamma = \kappa_a : \forall p . \forall x . p \ (A\ (B\ x))\Rightarrow p\ (A\ x) , \kappa_b : \forall p . \forall x . p \ (A\ x)\Rightarrow p\ (B\ x)$}
\end{center}
\noindent To represent the rewriting sequence above, we need to give a type to the function $f$ 
such that $\Gamma \vdash f \ \kappa_a \ \kappa_b : A\ x$. 
The most intuitive type we can assign to the corecursive function $f \alpha \ \beta = \alpha \ (f \ (\alpha \cdot \beta) \ \alpha)$ is the following: 

\begin{center}
(I)  $ \forall x. (\forall p_2. \forall y_2. p_2\ (A\ (B\ y_2))
  \Rightarrow p_2\ (A\ y_2)) \Rightarrow
  (\forall p_1. \forall y_1. p_1\ (A\ y_1) \Rightarrow p_1\ (B\ y_1))
  \Rightarrow A\ x$
\end{center}

\noindent Then we would have $\Gamma \vdash f \ x \ \kappa_a \ \kappa_b : A\ x$. Unfortunately this
will not be type checked by RSM (the resolution will fail). We need to perform abstraction on type (I), here we abstract the function symbol $B$ to a functional variable $b : * \Rightarrow *$,
and $A$ to a functional variable $a : * \Rightarrow *$, obtaining the following type for $f$.  

\begin{center}
(II) {\small $T \equiv \forall {a}. \forall b. \forall x.
  \underbrace{(\forall p. \forall y. p\ ({a}\ (\mgray{b}\ y))
  \Rightarrow p\ ({a}\ y)) \Rightarrow (\forall p. \forall y. p\ ({a}\ y) \Rightarrow p\ (\mgray{b}\ y))\Rightarrow a\ x}_{T'}$}
\end{center}

\noindent Note that the quantified variable $b$ in (II) is an existential variable. If $f$ is typable with (II), 
then we know that $\Gamma \vdash f \ A \ B \ x \ \kappa_a \ \kappa_b : A \ x$, which encodes the nonterminating
reduction starting from $A\ x$.
But RSM will fail again in this case, due to the appearance of the existential variable $b$. 

Ideally, the best way to deal with existential variables is by unification, we would need to replace rule (1) in RSM with the following:

\begin{center}
{  $\{(\Gamma, (\kappa | \alpha) \ e_1 \ ... \ e_n, A), \Phi \} \longrightarrow_a \{(\sigma \Gamma,e_1, \sigma T_1), ..., (\sigma \Gamma,e_1, \sigma T_n), \sigma \Phi\}$ if $\kappa | \alpha : \forall \underline{x}. T_1,..., T_n \Rightarrow B \in \Gamma$ with $B \sim_\sigma A$}
\end{center}

\noindent Here $B \sim_\sigma A$ means $A$ and $B$ are second-orderly unifiable by $\sigma$. And $\sigma \Gamma, \sigma \Phi$ means applying the substitution $\sigma$ to all the types in $\Gamma, \Phi$. But second-order unification is not decidable and we need a finite set of unifiers.
Thus we replace $B \sim_\sigma A$ with $B \mapsto_\sigma A$. %%  or $A \mapsto_\sigma B$, but we
%% quickly run into problems. For example, when $A \equiv p\ (G\ y)$ and $B \equiv F\ (a\ y)$, where $p, a$ are variables of kind $* \Rightarrow *$, in this case $A$ and $B$ are unifiable but not
%% mutually matchable. We decided to adopt the following heuristic for our application.  

\begin{definition}[Existential RSM (ERSM)]
\label{existential}
\

\noindent We replace (1) in Definition \ref{rsm} to the following (Keeping rules (2), (3), (4) unchanged):

(1') $\{(\Gamma, (\kappa | \alpha) \ e_1 \ ... \ e_n, A), \Phi \} \longrightarrow_a \{(\sigma \Gamma,e_1, \sigma T_1), ..., (\sigma \Gamma, e_n, \sigma T_n), \sigma \Phi\}$

if $\kappa | \alpha : \forall x_1. ... \forall x_m. T_1,..., T_n \Rightarrow B \in \Gamma$ with $B \mapsto_{\sigma} A$. 

  %% $\{(\Gamma, (\kappa | \alpha) \ e_1 \ ... \ e_n, A), \Phi \} \longrightarrow_a \{(\Gamma,e_1, \sigma T_1), ..., (\Gamma,e_1, \sigma T_n), \sigma \Phi\}$
  %% if $\kappa | \alpha : \forall \underline{x}. T_1,..., T_n \Rightarrow B \in \Gamma$ with $B \not \mapsto_\sigma A$ and $A \mapsto_\sigma B$.

 %% (1.1) $\{(\Gamma, \alpha, A), \Phi\} \longrightarrow_a \{\sigma \Phi\} $, if $\alpha : B \in \Gamma$, with $B \mapsto_\sigma A$. 

  %% \noindent Note that $\sigma \{(\Gamma, e, T), \Phi\} = \{(\sigma \Gamma, e, \sigma T), \sigma \Phi\}$, and all the free variables in the codomain of $\sigma$ are required to be
  %% declared in $\Gamma$, (1.1) would fail if this condition is not met. 
  
\end{definition}

Note that the formula $\forall x_1. ... \forall x_m. T_1,..., T_n \Rightarrow B$ in rule (1') may contain existential variables. 
The idea of this change is that by reordering the $(\Gamma, e, T)$ pairs, we give priority to
resolve the pair $(\Gamma, e, T)$ where the head of $T$ does not contain any existential
variables. If the $A$ in (1') does not contain existential variables, we can use rule (1') to eliminate the existential variables in $\forall x_1. ... \forall x_m. T_1,..., T_n \Rightarrow B$. This extension allows us to avoid using the undecidable second-order unification, and it is good enough to handle all of our examples involving existential variables\footnote{There is a well-known scope problem 
 \cite[Section 5]{dowek2001}, we show how to solve it for ERSM and prove the soundness of ERSM in Appendix \ref{scope}.}.

%% As the Rule (2) in RSM introduced eigenvariables, the formula $B$ in Rule (1.1) may contain existential variables, so once a substitution is found for the existential variables, we update the same existential variables shared by other formulas, hence $\sigma \Phi$. %% 

With the Definition \ref{existential}, we can obtain the following successful ERSM reduction,
where $\mu f . \lambda \alpha . \lambda \beta .\alpha \ (f \ (\lambda \alpha' . (\alpha \ (\beta\ \alpha'))) \ (\lambda \alpha' . \alpha \ \alpha'))$ is the long form of $f\ \alpha \ \beta = \alpha \ (f(\alpha \cdot \beta)\ \alpha)$.

\begin{center}
{
  $\{(\Gamma, \mu f . \lambda \alpha . \lambda \beta . \alpha \ (f \ (\lambda \alpha' . (\alpha \ (\beta\ \alpha')))\ (\lambda \alpha' . \alpha \ \alpha')) , T \} \longrightarrow_c$

  $\{([\Gamma, f : T], \lambda \alpha . \lambda \beta . \alpha \ (f \ (\lambda \alpha' . (\alpha \ (\beta\ \alpha'))) \ (\lambda \alpha' . \alpha \ \alpha')) , T )\} \longrightarrow_{\forall}$

  $\{([\Gamma, f : T], \lambda \alpha . \lambda \beta . \alpha \ (f \ (\lambda \alpha' . (\alpha \ (\beta\ \alpha'))) \ (\lambda \alpha' . \alpha \ \alpha')), [a_1/a, b_1/b, x_1/x]T' )\} \longrightarrow_i \{(\Gamma', \alpha \ (f \ (\lambda \alpha' . (\alpha \ (\beta\ \alpha'))) \ (\lambda \alpha' . \alpha \ \alpha')) , a_1\ x_1)\} \longrightarrow_a \{(\Gamma', f \ (\lambda \alpha' . (\alpha \ (\beta\ \alpha'))) \ (\lambda \alpha' . \alpha \ \alpha'), a_1\ (b_1 \ x_1)\} \mgray{\longrightarrow_a}$

  $\{(\Gamma', \lambda \alpha' . \alpha \ (\beta\ \alpha'), \forall p. \forall y. p\ (a_1\ (b_1\ (\mgray{b_2}\ y))) \Rightarrow p\ (a_1\ (b_1\ y))), \Phi \equiv (\Gamma', \lambda \alpha' . \alpha \ \alpha', (\forall p. \forall y. p\ (a_1\ (b_1\ y)) \Rightarrow p\ (\mgray{b_2}\ y)))\} \longrightarrow_{\forall}$

  $\{(\Gamma', \lambda \alpha' . \alpha \ (\beta\ \alpha'), p_2\ (a_1\ (b_1\ (\mgray{b_2}\ y_2))) \Rightarrow p_2\ (a_1\ (b_1\ y_2))), \Phi\} \longrightarrow_i $
  
  $\{([\Gamma', \alpha' : p_2\ (a_1\ (b_1\ (\mgray{b_2}\ y_2)))], \alpha \ (\beta\ \alpha'),  p_2\ (a_1\ (b_1 \ y_2)))), \Phi\} \longrightarrow_a$ 

  $\{([\Gamma', \alpha' : p_2\ (a_1\ (b_1\ (\mgray{b_2}\ y_2)))], \beta\ \alpha',  p_2\ (a_1\ (b_1 \ (b_1 \ y_2)))), \Phi\}\longrightarrow_a$

  $\{([\Gamma', \alpha' : p_2\ (a_1\ (b_1\ (\mgray{b_2}\ y_2)))], \alpha',  p_2\ (a_1\ (b_1 \ (a_1 \ y_2)))), \Phi\} \mgray{\longrightarrow_a}$
  
  $[(\lambda y. a_1 \ y)/\mgray{b_2}]\Phi \equiv \{(\Gamma', \lambda \alpha' . \alpha \ \alpha', \forall p. \forall y. p\ (a_1\ (b_1\ y)) \Rightarrow p\ (a_1\ y)) \} \longrightarrow_{\forall} $

    $\{(\Gamma', \lambda \alpha' . \alpha \ \alpha', p_3\ (a_1\ (b_1\ y_3)) \Rightarrow p_3\ (a_1\ y_3)) \} \longrightarrow_i $

  $\{([\Gamma', \alpha' : p_3\ (a_1 \ (b_1\ y_3))], \alpha\ \alpha', p_3\ (a_1\ y_3))\}  \longrightarrow_a $

  $\{([\Gamma', \alpha' : p_3\ (a_1 \ (b_1\ y_3))], \alpha', p_3\ (a_1\ (b_1\ y_3)))\} \longrightarrow_a \emptyset$

  %% $\{([\Gamma', p_3 : * \Rightarrow * , y_3 : *, \alpha' : p_3\ (a_1 \ (b_1\ (a_1\ y_3)))], \alpha', p_3\ (a_1\ (b_1\ (a_1\ y_3)))\} \longrightarrow_a \emptyset$
}
  
\end{center}
Note that $\Gamma' = \Gamma, f : T,  \alpha : \forall p. \forall y. p\ (a_1\ (b_1\ y)) \Rightarrow p\ (a_1\ y), \beta : \forall p. \forall y. p\ (a_1\ y) \Rightarrow p\ (b_1\ y)$. At the second $\longrightarrow_a$-step, by second-order
matching, variable $a$ is instantiated with $\lambda y. a_1\ (b_1\ y)$ for the type of $f$ and
the existential variable $b$ is instantiated with fresh variable $b_2$. At the fifth $\longrightarrow_a$-step, the existential variable $b_2$ is
instantiated with $\lambda y. a_1 \ y$, and there is a substitution for $b_2$ applying to $\Phi$. But RSM will not perform this substitution, as a result, RSM cannot resolve $\Phi$ to $\emptyset$.
\begin{comment}
This example shows that the use
of second-order existential variables enables us to provide functional certification for a wider range of
nonterminating reductions. The simple heuristic in Definition \ref{existential} allows us to handle the existential variables
by second-order matching. Once a substitution is found for an existential variable, we update the same existential variable shared by other formulas with the found value.
Thus we avoid the use of undecidable second-order unification. %, which is undecidable. 

%\knote{please check the below?}
We emphasize that ERSM should be seen as a heuristic, as it can produce evidence that fails to proof-check, see the attached appendix for an example.
%, i.e. as a practical way of helping automation 
%To make sure that the newly introduced heuristic does not affect soundness of our approach, 
%We want to note that the heuristic (Definition \ref{existential}) does not always work. 
%Note that we do not need to trust the heuristic,
This is why, we implement a separate proof 
checker to proof check the fully annotated corecursive evidence produced by the ERSM heuristic.
%For the purpose of our application, i.e. functional certification of nonterminating rewriting, the above
This heuristic
has proven to be practical in our applications, and in particular 
it generated correct annotated evidence that passed
the proof checker for many interesting examples (given in the attached appendix). 

\end{comment}

\section{Conclusion and Future Work}

\label{future}

We present a novel method to represent nonterminating reductions in $\mathbf{F}_2^\mu$, where the rewrite rules and first-order terms are modeled by types, and 
the nonterminations are modeled by the hereditary head normalizing evidence. We prove that the hereditary head normalizing evidence for a first-order term is faithful, i.e. it represents a nonterminating reduction. We also prove the hereditary head normalization property for $\mathbf{F}_2^\mu$ is decidable. To ease the representation process, we develop a type checking algorithm based on second-order matching, where fully annotated evidence can be generated from Curry-style evidence with only top-level type annotations. %% The connection between some nonlooping nonterminations and second-order quantifications appears to be new. 

% \textit{Remark}. We want to emphasis that proving nonlooping nontermination in general is very challenging,

\textbf{Future work}. %% We plan to investigate an alternative method of ensuring the hereditary
%% head normalization of the evidence in $\mathbf{F}_2^\mu$, i.e.
%We would like to explore methods of ensuring hereditary head normalization via typing. We
We would like to investigate the nonterminating reductions that are currently 
outside the scope of $\mathbf{F}_2^\mu$ and study the expressitivity of $\mathbf{F}_2^\mu$
in terms of representing nonterminations. %(e.g. as given in the example in Section \ref{discuss}) 
The RSM/ERSM type checking algorithm is not very flexible. For example
the Curry style evidence currently has to be in long form. We plan to
 relax this restriction.
%% We would also like to explore the potential applications of functional certification for 
%% safety property analysis. % to handle the forms of types that is not of the form $T :: = \forall \underline{x} . T_1 \Rightarrow ... \Rightarrow T_n \Rightarrow A$.

%\newpage

\section*{Acknowledgement}
I would like to thank Tom Schrijvers for coming up with Example \ref{tom} and showing me a solution in Haskell using type family (See Fu et. al. \cite{fu2016}), at a time when I thought this whole thing is impossible. I also like to thank Ekaterina Komendantskaya for many helpful
discussions, which leads me to consider the automation aspect, eventually I discover that quantification over higher-order variables leads to another solution for Example \ref{tom} without using type family, hence this paper. Reviewer 1 from FSCD 2016 discovered an error
in an ealier version of the paper, which leads to a more rigid formulation of $\mathbf{F}_2^\mu$.
Reviewer A from POPL 2017 suggests a possible simplification of productivity checking by mapping
$\mathbf{F}_2^\mu$ to $\lambda$-Y, which I carried out in this paper, and it greatly simplifies and strengthens the paper. Leibniz representation in this paper is inspired by Stump and Sch{\"u}rmann \cite{stump2005logical}'s treatment on rewriting and
Girard's recent criticism about Leibniz equality \footnote{J.Y. Girard, \textit{Transcendental syntax III: equality}}. I would also like to thank the School of Computing at University of Dundee, and my mother Chen Xingzhen for generously providing a working space for me when I was in
transitions between Postdocs. 

% Acknowledgments, if needed.

% We recommend abbrvnat bibliography style.

%\bibliographystyle{plain}
\bibliography{nonterm}
% The bibliography should be embedded for final submission.

% \begin{thebibliography}{}
% \softraggedright

% \bibitem[Smith et~al.(2009)Smith, Jones]{smith02}
% P. Q. Smith, and X. Y. Jones. ...reference text...

% \end{thebibliography}

\newpage

\appendix

\section{Proof of Theorem \ref{sub:kind}}

\begin{theorem}
 If $\Delta \vdash T : k$ and $T \to_o T'$, then $\mathrm{FV}(T) = \mathrm{FV}(T')$ and $\Delta \vdash T' : k$  
\end{theorem}
\begin{proof}
  By induction on the derivation of $\Delta \vdash T : k$. 

  \noindent \textbf{Case}. 
  
\

 \infer{\Delta \vdash x | F : K}{(x | F :  K) \in \Delta}    

\

Obvious.
  
  \noindent \textbf{Case}. 
  
%  \infer{\Delta \vdash x | F : K}{(x | F :  K) \in \Delta}    

  \
  
\infer{\Delta \vdash \lambda x. T : * \Rightarrow K}{\Delta, x : * \vdash T : K & x \in \mathrm{FV}(T)}

\

\noindent We have $T \to_o T'$. By IH, we have $\Delta, x : * \vdash T' : K$ and 
$\mathrm{FV}(T) = \mathrm{FV}(T')$. Thus $x \in \mathrm{FV}(T')$. So $\Delta \vdash \lambda x. T' : * \Rightarrow K$.

\noindent \textbf{Case}. 

\

\infer{\Delta \vdash (\lambda x.T_2)\ T_1 : K}{\Delta \vdash T_1 : * & \Delta \vdash \lambda x.T_2 : * \Rightarrow K}

\

\noindent We have $(\lambda x.T_2)\ T_1 \to_o [T_1/x]T_2$. Since $\Delta \vdash \lambda x.T_2 : * \Rightarrow K$, by inversion we know that $\Delta, x : * \vdash T_2 : K$ and $x \in \mathrm{FV}(T_2)$. 
So $\mathrm{FV}((\lambda x.T_2)\ T_1) = \mathrm{FV}([T_1/x]T_2)$ and $\Delta \vdash [T_1/x]T_2 : K$. 

\

\noindent \textbf{Case}. 

\

\infer{\Delta \vdash  \forall x . T : o }{\Delta, x : K \vdash T :  o | *}

\

Suppose $\forall x . T\to_o \forall x . T'$ by $T\to_o T'$. By IH, $\Delta, x : K \vdash T' :  o | * $ and $\mathrm{FV}(T) = \mathrm{FV}(T')$. Thus $\Delta \vdash  \forall x . T' : o$ and $\mathrm{FV}(\forall x . T) = \mathrm{FV}(\forall x .T')$.

\

All the other cases are similar. 
%% \infer{\Delta \vdash  T' \Rightarrow T : o }{\Delta \vdash T : \phi(T) & \Delta \vdash T' : \phi(T')}

\end{proof}

\section{Proof of Theorem \ref{prop:kinding}}
\begin{theorem}
  \
  
  \begin{enumerate}
  \item If $\Delta \vdash T : o$, then $T$ is of the form $\forall x. T'$ or $T_1 \Rightarrow T_2$.
    \item If $\Delta \vdash T : *^n \Rightarrow *$, then the normal form of $T$ is second-order.
  \end{enumerate}
\end{theorem}
\begin{proof}

  \noindent (1) Obvious. 
  
  \noindent (2). By induction on the derivation of $\Delta \vdash T : *^n \Rightarrow *$. 
  
\noindent \textbf{Case}.

\

  \infer{\Delta \vdash x | F : K}{(x | F :  K) \in \Delta}    

  \
  
  \noindent Obvious. 

  \noindent \textbf{Case}.
  
  \
  
\infer{\Delta \vdash T_2 \ T_1 : K}{\Delta \vdash T_1 : * & \Delta \vdash T_2 : * \Rightarrow K}

\

\noindent We need to show the normal form of $T_2 \ T_1$ is second-order. By IH, we know 
the normal form of $T_1, T_2$ are second-order, moreover, $T_1$ is flat since $\Delta \vdash T_1 : *$. Suppose $T_2 \equiv F$ or $T_2 \equiv x$, then
by definition we know $T_2\ T_1$ is second-order. Suppose $T_2 \equiv \lambda x.T'$, where 
$x \in \mathrm{FV}(T')$ and $T'$ is second-order. Then $(\lambda x.T')\ T_1 \to_o [T_1/x]T'$ and $[T_1/x]T'$
is second-order.

  \noindent \textbf{Case}.
  
  \

\infer{\Delta \vdash \lambda x. T : * \Rightarrow K}{\Delta, x : * \vdash T : K & x \in \mathrm{FV}(T)}

\

\noindent Let  $[T]$ be the normal form of $T$. By IH, we know that $[T]$ is second-order. By Theorem \ref{sub:kind}, 
we know that $x \in \mathrm{FV}([T])$. Thus $\lambda x. [T]$ is second-order. 
\end{proof}

\section{Proof of Theorem \ref{confluence}}

\begin{theorem}

 $\leadsto_{\beta\mu\tau o}$ and
    $\leadsto_{h\tau o}$ are confluent, and $\leadsto_\tau$ is strongly normalizing.  
\end{theorem}
\begin{proof}
  Note that $\leadsto_\tau$ commutes with $\leadsto_o$, $\leadsto_h$ and $\leadsto_{\beta\mu}$. Also $\leadsto_o$ commutes with $\leadsto_h$ and $\leadsto_{\beta\mu}$. 
  Thus it is enough
to show that $\leadsto_h$ and $\leadsto_{\beta\mu}$ are confluent. For $\leadsto_{h}$, we just need to check 
$\mathcal{H}_1[(\lambda x . (\mathcal{H}_2[\mu \alpha. e']))\ e]$, as it is the only critical pair. We know that:

\begin{center}
{\footnotesize
  \noindent $\mathcal{H}_1[(\lambda \alpha . (\mathcal{H}_2[\mu
  \beta. e']))\ e] \leadsto_h
  \mathcal{H}_1[([e/\alpha]\mathcal{H}_2)[\mu \beta. [e/\alpha]e']]
  \leadsto_h \mathcal{H}_1[([e/\alpha]\mathcal{H}_2)[[\mu
  \beta. [e/\alpha]e'/\beta][e/\alpha] e']] $

\

  \noindent $\mathcal{H}_1[(\lambda \alpha . (\mathcal{H}_2[\mu
  \beta. e']))\ e] \leadsto_h \mathcal{H}_1[(\lambda \alpha
  . \mathcal{H}_2[[\mu \beta. e'/\beta]e'])\ e] \leadsto_h
  \mathcal{H}_1[([e/\alpha]\mathcal{H}_2)[[\mu \beta. [e/\alpha]
  e'/\beta] [e/\alpha]e']]$
}
\end{center}
Thus $\leadsto_h$ is confluent. For the confluence of $\leadsto_{\beta\mu}$, we refer
to the existing literature (e.g. \cite[\S 7.1]{ariola1997}). Finally, $\leadsto_\tau$ 
is strongly normalizing because the number of $\leadsto_\tau$-redex is strictly decreasing.
\end{proof}

\section{Proof of Theorem \ref{sub-red}}

\begin{theorem}[Inversion]
  \label{app:inversion}
\

  \begin{enumerate}
  \item If $\Gamma \vdash x : T$, then there exists $(x : T') \in \Gamma$ and $T \leftrightarrow_o^* T'$.
  \item If $\Gamma \vdash \kappa : T$, then there exists $(\kappa : T') \in \Gamma$ and $T \leftrightarrow_o^* T'$.
\item If $\Gamma \vdash \lambda \alpha . e : T$, then $\Gamma, \alpha : T_1 \vdash e : T_2$ and $T_1 \Rightarrow T_2 \leftrightarrow_o^* T$.

\item  If $\Gamma \vdash e\ e' : T$, then $\Gamma \vdash e : T_1 \Rightarrow T_2$, $\Gamma \vdash e' : T_1$
  and $T_2 \leftrightarrow_o^* T$.

\item If $\Gamma \vdash \lambda x . e : T$, then $\Gamma \vdash e : T'$, $x \notin \mathrm{FV}(\Gamma)$ and $\forall x. T' \leftrightarrow_o^* T$.

\item  If $\Gamma \vdash e\ T_1 : T$, then $\Gamma \vdash e : \forall x . T'$
  and $[T_1/x]T' \leftrightarrow_o^* T$.

\item If $\Gamma \vdash \mu \alpha . e : T$, then $\Gamma, \alpha : T' \vdash e : T'$ and $ T' \leftrightarrow_o^* T$.
  \end{enumerate}
\end{theorem}
\begin{proof}
  By induction on derivation. 
\end{proof}
\begin{lemma}
\label{subst-lemma}
\

\begin{enumerate}
\item $\Gamma, \alpha : T \vdash e : T'$ and $\Gamma \vdash e' : T$,
  then $\Gamma \vdash [e'/\alpha]e : T'$.
\item $\Gamma \vdash e : T'$, 
  then $[T_1/x]\Gamma \vdash [T_1/x]e : [T_1/x]T'$.

\end{enumerate}
\end{lemma}
\begin{proof}
  By induction on the derivation.  
\end{proof}
\begin{theorem}
  If $\Gamma \vdash e : T$ and $e \leadsto_{h\tau o} e'$, then $\Gamma \vdash e' : T$.   
\end{theorem}

\begin{proof}
  By induction on the derivation of $\Gamma \vdash e : T$. 

\noindent \textbf{Case}. 

\

\infer[(\textsc{Mu})]{\Gamma \vdash \mu \alpha . e : T}{\Gamma, \alpha : T \vdash e : T}

\

\noindent We know $\mu \alpha . e \leadsto_h [\mu \alpha . e/\alpha]e$. By lemma \ref{subst-lemma} (1), we know 
that $\Gamma \vdash [\mu \alpha . e/\alpha]e : T$.

\noindent \textbf{Case}. 

\

\infer[(\textsc{App})]{\Gamma \vdash (\lambda \alpha.e)\ e_1 : T}{\Gamma \vdash e_1 : T' & \Gamma \vdash \lambda \alpha.e : T' \Rightarrow T}

\

\noindent Suppose $(\lambda \alpha.e)\ e_1 \leadsto_h [e_1/\alpha]e$. By Theorem \ref{app:inversion} (4), we have
$\Gamma, \alpha : T_1 \vdash e : T_2$ and $T_1 \Rightarrow T_2 \leftrightarrow_o^* T' \Rightarrow T$. Since 
$\to_o$ is confluent, we have $T_1 \leftrightarrow_o^* T'$ and $T_2 \leftrightarrow_o^* T$. Thus
$\Gamma \vdash e_1 : T_1$. By Lemma \ref{subst-lemma} (1), we know $\Gamma \vdash [e_1/\alpha]e : T_2$. Thus 
$\Gamma \vdash [e_1/\alpha]e : T$.

\noindent \textbf{Case}. 

\

\infer[(\textsc{Inst})]{\Gamma \vdash (\lambda x.e) \ T': [T'/x]T}{\Gamma \vdash \lambda x.e : \forall x : K . T }

\

\noindent Suppose that $(\lambda x.e)\ T' \leadsto_\tau [T'/x]e$. By Theorem \ref{app:inversion} (5), we have
$\Gamma \vdash e : T_1$, $x \notin \mathrm{FV}(\Gamma)$ and $\forall x.T_1 \leftrightarrow^*_o \forall x . T$. By Lemma \ref{subst-lemma} (2), we have $\Gamma \vdash [T'/x]e :[T'/x]T_1$. Since $\forall x.T_1 \leftrightarrow^*_o \forall x . T$ implies $[T'/x]T_1 \leftrightarrow^*_o [T'/x]T$, we have $\Gamma \vdash [T'/x]e :[T'/x]T$.

\

\noindent Suppose that $(\lambda x.e)\ T' \leadsto_o (\lambda x.e)\ T''$ with $T' \to_o T''$.
So by \textsc{App} rule, we have $\Gamma \vdash (\lambda x.e)\ T'' : [T''/x]T$. By \textsc{Conv} rule, we have $\Gamma \vdash (\lambda x.e)\ T'' : [T'/x]T$.

\

\noindent For all the other cases are easy. 
\end{proof}

\section{Proof of Theorem \ref{inversion}}
\begin{lemma}
Suppose $\Gamma_{\mathcal{R}} \vdash e : t$ for some first-order term $t$ and $e$ is head normalizing. We have $e \leadsto_{h\tau}^* \kappa \ (\lambda x . C[x,..., x]) \ t_1 ...\ t_n\  e'$ for some $\kappa :  \forall p . \forall \underline{x}. p \ r \Rightarrow p\ l \in \Gamma_{\mathcal{R}}$.
Furthermore, we have $\Gamma_{\mathcal{R}} \vdash e' :  C[\sigma r,..., \sigma r]$ and $C[\sigma l,...,\sigma l] = t$, where $\mathrm{codom}(\sigma) = \{t_1,..., t_n\}$ and $\mathrm{dom}(\sigma) = \mathrm{FV}(l)$. 
\end{lemma}
\begin{proof}
  Since $e$ is head normalizing and $\Gamma_{\mathcal{R}} \vdash e : t$, its head normal form must be of the $\kappa \ T \ T_1 ...\ T_n\  e'$ for some $\kappa :  \forall p . \forall \underline{x}. p \ r \Rightarrow p\ l \in \Gamma_{\mathcal{R}}$. By subject reduction (Theorem \ref{sub:kind}, Theorem \ref{sub-red}), we have $\Gamma_{\mathcal{R}} \vdash \kappa \ T \ T_1 ...\ T_n\  e' : t$. By inversion Theorem \ref{thm:inversion} (1) on $\Gamma_{\mathcal{R}} \vdash \kappa \ T \ T_1 ...\ T_n\  e' : t$, we know that $\Gamma_{\mathcal{R}}\vdash \kappa \ T \ T_1 ...\ T_n : T_1' \Rightarrow T_2'$, $\Gamma_{\mathcal{R}}\vdash e' : T_1'$ and $T_2' \leftrightarrow_o t$. By inversion Theorem \ref{thm:inversion} (2) on $\Gamma_{\mathcal{R}} \vdash \kappa \ T \ T_1 ...\ T_n : T_1' \Rightarrow T_2'$,
  we have $\sigma (p\ r) \Rightarrow \sigma (p\ l) \leftrightarrow^*_o T_1' \Rightarrow T_2'$, where $\sigma = [T/p, T_1/x_1,..., T_n/x_n]$. Since we are working with
  well-kinded types, we know that
$\Gamma_{\mathcal{R}} \vdash T : * \Rightarrow *$ and $\Gamma_{\mathcal{R}} \vdash T_i : *$ for all $i$. By Theorem \ref{prop:kinding}, we know $T = \lambda x . C[x,...,x]$ and $T_i$ is flat for all
  $i$. By confluence of $\leftrightarrow_o$, we have $\sigma (p\ r) \leftrightarrow_o^* T_1'$ and $\sigma (p\ l) \leftrightarrow_o^* T_2' \leftrightarrow_o^* t$. Thus $\sigma (p\ l) \equiv [T/p, T_1/x_1,..., T_n/x_n] (p\ l) \equiv (\lambda x. C[x,...,x])\ (\sigma l) \to_o^* t$. So $C[\sigma l,..., \sigma l] = t$. Since $\sigma (p\ r) \leftrightarrow_o^* T_1'$, we have $\Gamma_{\mathcal{R}}\vdash e' : C[\sigma r,..., \sigma r]$. 
\end{proof}

\section{Mapping $\mathbf{F}_{2}^\mu$ to $\lambda$-Y}
\label{lambda-y}

\begin{definition}[$\lambda$-Y calculus]

  \
  
  \textit{$\lambda$-Y terms} $e ::=  \alpha \ | \ \kappa ~\mid~ \lambda \alpha . e ~\mid~ e \ e' ~\mid~ \mu \alpha . e$

    \textit{$\lambda$-Y types} $T ::= B \ |\ T \Rightarrow T'$
    
    \textit{$\lambda$-Y environment} $\Gamma ::=  \cdot ~\mid~ \alpha : T, \Gamma \ | \ \kappa : T$
    
\end{definition}
Note that $B$ denotes a constant type in $\lambda$-Y. 

\begin{definition}[Typing of $\lambda$-Y]
\

{\scriptsize
\begin{tabular}{ll}
\\
\infer{\Gamma \vdash \alpha | \kappa : T}{(\alpha | \kappa :  T) \in \Gamma}    
&

\infer[(\textsc{App})]{\Gamma \vdash e_2\ e_1 : T}{\Gamma \vdash e_1 : T' & \Gamma \vdash e_2 : T' \Rightarrow T}
\\
\\
\infer[(\textsc{Lam})]{\Gamma \vdash \lambda \alpha.e : T' \Rightarrow T}{\Gamma, \alpha : T' \vdash e : T}
&
\infer[(\textsc{Mu})]{\Gamma \vdash \mu \alpha . e : T}{\Gamma, \alpha : T \vdash e : T}
  \end{tabular}  
}
\end{definition}

\begin{definition}
  We define a function $\theta$ that maps $\mathbf{F}_{2}^\mu$ types to $\lambda$-Y types.
\

  $\theta(F) = B$

  $\theta(x) = B$ 

  $\theta(\lambda x. T) = \theta(T)$

  $\theta(T\ T') = \theta(T)$

  $\theta(T \Rightarrow T') = \theta(T) \Rightarrow \theta(T')$

  $\theta(\forall x.T) = \theta(T)$
  
\end{definition}

\begin{lemma}
  \label{term}
  If $\Delta \vdash T : K$ in $\mathbf{F}_{2}^\mu$, then $\theta(T) = B$. 
\end{lemma}
\begin{proof}
By induction on the derivation of $\Delta \vdash T : K$.   
\end{proof}
\begin{lemma}
  \label{y:subst}
  If $\Delta \vdash T' : K$ in $\mathbf{F}_{2}^\mu$, then $\theta([T'/x]T) \equiv \theta(T)$
  for any $T$ in $\mathbf{F}_{2}^\mu$.
\end{lemma}
\begin{proof}
Using Lemma \ref{term} and induction on the structure of $T$. 
\end{proof}

\begin{lemma}
  \label{y:oe}
  If $T_1 \leftrightarrow_o^* T_2$ and $\Delta \vdash T_1|T_2 : k$ in $\mathbf{F}_{2}^\mu$, then $\theta(T_1) \equiv \theta(T_2)$.
\end{lemma}
\begin{proof}
  By induction on the derivation of $T_1 \leftrightarrow_o^* T_2$.
\end{proof}
\begin{definition}

  \
  
  $\theta(.) = .$
  
  $\theta(\Gamma, \alpha : T) = \theta(\Gamma), \alpha : \theta(T)$
  
  $\theta(\Gamma, \kappa : T) = \theta(\Gamma), \kappa : \theta(T)$

  %% $\theta(\Gamma, x : K) = \theta(\Gamma)$

  %%  $\theta(\Gamma, F : K) = \theta(\Gamma)$
\end{definition}

%% \begin{lemma}
%%   If $\Gamma \vdash T : *|o $ in $\mathbf{F}_{2}^{\mu}$, then
%%   $\theta(T)$ is a $\lambda$-Y type.
%% \end{lemma}

\begin{theorem}
  If $\Gamma \vdash e : T$ and $\Delta \vdash T : *|o $ in $\mathbf{F}_{2}^{\mu}$, then
  $\theta(\Gamma)\vdash |e| : \theta(T)$ in $\lambda$-Y.
\end{theorem}
\begin{proof}
  By induction on derivaton of $\Gamma \vdash e : T$ in $\mathbf{F}_{2}^{\mu}$. 
  \begin{itemize}
  \item Case:

    \
    
    \begin{tabular}{l}
      \infer{\Gamma \vdash \alpha | \kappa : T}{(\alpha | \kappa :  T) \in \Gamma}    
    \end{tabular}
    
    \

    We just need to show $\theta(\Gamma) \vdash \alpha | \kappa : \theta(T)$ in $\lambda$-Y, which we know
    is the case by definition of $\theta(\Gamma)$. 

    \item Case: 

      \
      
      \begin{tabular}{l}
        \infer[(\textsc{App})]{\Gamma \vdash e_2\ e_1 : T}{\Gamma \vdash e_1 : T' & \Gamma \vdash e_2 : T' \Rightarrow T}
      \end{tabular}
\

We need to show $\theta(\Gamma) \vdash |e_2\ e_1| : \theta(T)$ in $\lambda$-Y. By induction,
we know that $\theta(\Gamma) \vdash |e_1| : \theta(T')$ and $\theta(\Gamma) \vdash |e_2| : \theta(T') \Rightarrow \theta(T)$ in $\lambda$-Y. Thus we have  $\theta(\Gamma) \vdash |e_2|\ |e_1| : \theta(T)$.

    \item Case: 

      \
      
      \begin{tabular}{l}
        \infer[(\textsc{Lam})]{\Gamma \vdash \lambda \alpha.e : T' \Rightarrow T}{\Gamma, \alpha : T' \vdash e : T}
      \end{tabular}

      \

      We need to show $\theta(\Gamma) \vdash \lambda \alpha.|e| : \theta(T') \Rightarrow \theta(T)$ in $\lambda$-Y. By induction, we know that $\theta(\Gamma), \alpha : \theta(T') \vdash |e| : \theta(T)$ in $\lambda$-Y.
      
    \item Case: 

      \
      
      \begin{tabular}{l}
        \infer[(\textsc{Mu})]{\Gamma \vdash \mu \alpha . e : T}{\Gamma, \alpha : T \vdash e : T}
      \end{tabular}

      \

      We need to show $\theta(\Gamma) \vdash \mu \alpha . |e| : \theta(T)$ in $\lambda$-Y. By induction, we know that $\theta(\Gamma), \alpha : \theta(T) \vdash |e| : \theta(T)$ in $\lambda$-Y.
      
    \item Case: 

      \
      
      \begin{tabular}{l}
        \infer[(\textsc{Abs})]{\Gamma \vdash \lambda x . e: \forall x : K . T}{\Gamma  \vdash e :  T & x \notin \mathrm{FV}(\Gamma)}
      \end{tabular}

      \

      We need to show $\theta(\Gamma) \vdash |e| : \theta(T)$ in $\lambda$-Y, which is the
      case by induction. %, we know that $\theta(\Gamma) \vdash |e| : \theta(T)$ in $\lambda$-Y.

    \item Case: 

      \
      
      \begin{tabular}{l}
        \infer[(\textsc{Inst})]{\Gamma \vdash e \ T': [T'/x]T}{\Gamma \vdash e : \forall x : K . T}
      \end{tabular}

      \

      We need to show $\theta(\Gamma) \vdash |e| : \theta([T'/x]T)$ in $\lambda$-Y.
      By induction, we know that $\theta(\Gamma) \vdash |e| : \theta(T)$. By Lemma \ref{y:subst},
      we know that $\theta([T'/x]T) \equiv \theta(T)$.
      
    \item Case: 

      \
      
      \begin{tabular}{l}
        \infer[(\textsc{Conv})]{\Gamma \vdash e : T'}{\Gamma \vdash e : T & T \leftrightarrow^*_o T'}
      \end{tabular}

      \

      We need to show $\theta(\Gamma) \vdash |e| : \theta(T')$ in $\lambda$-Y.
      By induction, we know that $\theta(\Gamma) \vdash |e| : \theta(T)$. 
      By Lemma \ref{y:oe}, we know that $\theta(T') \equiv \theta(T)$.
  \end{itemize}
\end{proof}

\section{Proof of Theorem \ref{thm:sound}}
\begin{lemma}
  \label{res:sound}
If $\{(\Gamma_1, e_1, T_1),..., (\Gamma_n, e_n, T_n)\} \longrightarrow^* \emptyset$, then there exists an evidence $e'_1,..., e'_n$ such that $\Gamma_i \vdash e'_i : T_i$ and $|e'_i| = e_i$ for all $i$.  
\end{lemma}
\begin{proof}
  By induction on the length of
  $\{(\Gamma_1, e_1, T_1),..., (\Gamma_n, e_n, T_n)\}$ $\longrightarrow^* \emptyset$.
  \begin{itemize}
  \item Case $\{(\Gamma, \alpha | \kappa, A)\} \longrightarrow_a \emptyset$.
    
    In this case $\alpha | \kappa : \forall \underline{x}. B \in \Gamma$ and $B \mapsto_\sigma A$. Since $\forall \underline{x}. B$ does not contain existential variables, by \textsc{Inst}, we have $\Gamma \vdash (\alpha | \kappa) \ \underline{T} : A$, where $\{\underline{T}\} = \mathrm{codom}(\sigma)$ and $|(\alpha | \kappa) \ \underline{T}| = \alpha | \kappa$.

  \item Case

    $\{(\Gamma, (\alpha | \kappa)\ e''_1 \ ...  e_m'', A), (\Gamma_1, e_1, T_1), ..., (\Gamma_n, e_n, T_n)\} \longrightarrow_a $

    $\{(\Gamma,e_1'', \sigma T_1'), ..., (\Gamma,e_m'', \sigma T_m'), (\Gamma_1, e_1, T_1), ..., (\Gamma_n, e_n, T_n)\} \longrightarrow^* \emptyset$, where $\kappa | \alpha : \forall \underline{x}. T_1',..., T_n' \Rightarrow B \in \Gamma$ with $B \mapsto_\sigma A$. 
      
      By IH, we know that $\Gamma \vdash e_1''' :  \sigma T_1', ..., \Gamma \vdash e_m'' : \sigma T_m', \Gamma_1 \vdash e_1' : T_1, ..., \Gamma_n \vdash e_n' : T_n$ and $|e_1'''| = e_1'',..., |e_m'''| = e_m'',|e_1'| = e_1, ..., |e_n'| = e_n$. Let $\mathrm{codom}(\sigma) = \underline{T}$, since $\forall \underline{x}. T_1',..., T_n' \Rightarrow B$ does not contain existential variables, we have $\Gamma \vdash (\alpha | \kappa) \ \underline{T}: \sigma T_1',..., \sigma T_n' \Rightarrow \sigma B $. Thus
      $\Gamma \vdash (\alpha | \kappa) \ \underline{T}\ e_1'''\  ... \ e_m'''  : \sigma B$. By 
      \textsc{Conv}, we have $\Gamma \vdash (\alpha | \kappa) \ \underline{T}\ e_1'''\  ... \ e_m'''  : A$. Moreover, $|(\alpha | \kappa) \ \underline{T}\ e_1'''\  ... \ e_m'''| = (\alpha | \kappa) \ |e_1'''|\  ... \ |e_m'''| = (\alpha | \kappa) \ e_1''\  ... \ e_m''$.
      
      \item Case 
{
    $\{(\Gamma, \lambda \alpha_1. ... \lambda \alpha_n. e, T_1, ..., T_n \Rightarrow A), (\Gamma_1, e_1, T_1), ..., (\Gamma_l, e_l, T_l) \} $
$\longrightarrow_i \{([\Gamma, \alpha_1 : T_1, ..., \alpha_n : T_n], e, A), (\Gamma_1, e_1, T_1), ..., (\Gamma_l, e_l, T_l) \} \longrightarrow^* \emptyset$}

By IH, we have $\Gamma, \alpha_1 : T_1, ..., \alpha_n : T_n \vdash e' : A, \Gamma_1 \vdash e_1' : T_1, ...,\Gamma_l \vdash e_l' : T_l$
with $|e'| = e, |e_1'| = e_1, ..., |e_l'| = e_l$. Thus by \textsc{Lam} rule, we have $\Gamma \vdash  \lambda \alpha_1 .... \lambda \alpha_n . e' : T_1, ..., T_n \Rightarrow A$ and $|\lambda \alpha_1 .... \lambda \alpha_n . e'| = \lambda \alpha_1. ... \lambda \alpha_n. e$.

      \item Case 
{
  $\{(\Gamma, e, \forall x_1. ...\forall x_m . T), (\Gamma_1, e_1, T_1), ..., (\Gamma_l, e_l, T_l) \} \longrightarrow_{\forall} $

  $\{(\Gamma, e, T), (\Gamma_1, e_1, T_1), ..., (\Gamma_l, e_l, T_l) \} \longrightarrow^* \emptyset$}

By IH, we have $\Gamma \vdash e' : A, \Gamma_1 \vdash e_1' : T_1, ...,\Gamma_l \vdash e_l' : T_l$
with $|e'| = e, |e_1'| = e_1, ..., |e_l'| = e_l$. Since $\{x_1,..., x_m\} \cap \mathrm{FV}(\Gamma) = \emptyset$, by \textsc{Abs} rules, we have $\Gamma \vdash \lambda x_1 .... \lambda x_m . e' : \forall x_1 ... . \forall x_m . T$ and $|\lambda x_1 .... \lambda x_m . e'| = e$.

\item Case  $\{(\Gamma, \mu \alpha . e, T), (\Gamma_1, e_1, T_1), ..., (\Gamma_n, e_n, T_n)\} \longrightarrow_c\{([\Gamma, \alpha : T], e, T), (\Gamma_1, e_1, T_1), ..., (\Gamma_n, e_n, T_n)\} $

  $\longrightarrow^* \emptyset$
    
    By IH, we know that $\Gamma, \alpha : T \vdash e' : T, \Gamma_1 \vdash e_1' : T_1, ..., \Gamma_n \vdash e_n' : T_n$ and $|e_i'| = e_i$ for all $i$. By \textsc{Mu} rule, we have $\Gamma \vdash \mu \alpha . e' : T$. Thus $|\mu \alpha. e'| = \mu \alpha. e$.

  \end{itemize}
\end{proof}

\section{Examples in the Paper}
\label{app:expaper}

In this section we show how to represent nonterminations for all the examples
in the paper using the prototype FCR (for Functional Certification of Rewriting), the
prototype is available at \url{https://github.com/Fermat/FCR}. 
It tries to generate typable $\mathbf{F}_{2}^{\mu}$ evidence from the corecursive equations
and the type declarations. 

\subsection{Example in Section \ref{typecheck}}
The following is the input file for FCR. 
{
\begin{verbatim}
A : forall p x y . p (D x (S y)) => p (D (S x) y)
B : forall p y . p (D (S y) Z) => p (D Z y)

g : forall d . 
     (forall p x y . p (d x (S y)) => p (d (S x) y)) => 
     (forall p y . p (d (S y) Z) => p (d Z y)) => 
     d Z Z

g a1 a2 = a2 (a1 (g (\ v . a1 v) (\ v . a2 (a1 v))))

e : D Z Z
e = g (\ v . A v) (\ v . B v)
\end{verbatim}
}

The capitalized words for FCR are intended to denote both type and evidence constant, uncapitalized words are intended to denote both type and evidence variables. In the definition of corecursive
function \texttt{g}, ``\texttt{\textbackslash}'' denotes the $\lambda$ binder, its type declaration
is discussed in the paper. FCR currently uses long normal form to make variable instantiation, so we
have to use (I) instead of (II). 

\begin{center}
{
  \noindent (I) \texttt{g a1 a2 = a2 (a1 (g (\textbackslash v . a1 v)
    (\textbackslash v . a2 (a1 v))))}

  \noindent (II) \texttt{g a1 a2 = (a2 . a1) (g a1 (a2 . a1))}
}
\end{center}

\noindent Evidence such as $\mu f. \lambda a. e$ is 
represented as equation \texttt{f a = e}, so there is no explicit $\mu$ binder in the input file. The corecursive evidence for \texttt{D Z Z} is \texttt{e}. 
The following is the output by the type checker. 
{
\begin{verbatim}
rewrite rules
kinds
D : * => * => *
S : * => *
Z : *
axioms
A : forall p x y . p (D x (S y)) => p (D (S x) y)
B : forall p y . p (D (S y) Z) => p (D Z y)
proof declarations
g : forall d .
      (forall p x y . p (d x (S y)) => p (d (S x) y))
      =>
        (forall p y . p (d (S y) Z) => p (d Z y)) => d Z Z =
\ a1 a2 . a2 (a1 (g (\ v . a1 v) (\ v . a2 (a1 v))))
e : D Z Z =
g (\ v . A v) (\ v . B v)
lemmas
e : D Z Z =
  g (\ m1' m2' . D m1' m2')
    (\ p1' x2' y3' (v : p1' (D x2' (S y3'))) .
         A (\ m1' . p1' m1') x2' y3' v)
    (\ p7' y8' (v : p7' (D (S y8') Z)) . B (\ m1' . p7' m1') y8' v)
g : forall d .
      (forall p x y . p (d x (S y)) => p (d (S x) y))
      =>
        (forall p y . p (d (S y) Z) => p (d Z y)) => d Z Z =
  \ d0'
    (a1 : forall p x y . p (d0' x (S y)) => p (d0' (S x) y))
    (a2 : forall p y . p (d0' (S y) Z) => p (d0' Z y)) .
      a2 (\ x1' . x1') Z
        (a1 (\ x1' . x1') Z Z
           (g (\ m1' m2' . d0' m1' (S m2'))
              (\ p7' x8' y9' (v : p7' (d0' x8' (S (S y9')))) .
                   a1 (\ m1' . p7' m1') x8' (S y9') v)
              (\ p13' y14' (v : p13' (d0' (S y14') (S Z))) .
                   a2 (\ m1' . p13' m1') (S y14')
                     (a1 (\ m1' . p13' m1') (S y14') Z v))))
steps
automated proof reconstruction success!
\end{verbatim}
}

\noindent The \texttt{lemmas} section contains the annotated evidence. All variables generated by FCR
are variables end with `` ' ''. All lambda-bound evidence variables are annotated with the type information. This is needed for decidable proof checking, we do not need to annotate lambda-bound type variables. The annotated evidence generated by our type checker is checked by a separate $\mathbf{F}_{2}^{\mu}$ proof checker. 

We can translated the input file into the following Haskell code, but it will not pass Haskell's
type checker. 

{
\begin{verbatim}
data D :: * -> * -> *
data S :: * -> *
data Z :: *
a :: forall p x y .  p (D x (S y)) -> p (D (S x) y)
a = undefined
b :: forall p y . p (D (S y) Z) -> p (D Z y)
b = undefined
g :: forall d . 
     (forall p x y . p (d x (S y)) -> p (d (S x) y)) -> 
     (forall p y . p (d (S y) Z) -> p (d Z y)) -> 
     d Z Z
g a1 a2 = a2 (a1 (g (\ v -> a1 v) (\ v -> a2 (a1 v))))

e :: D Z Z
e = g (\ v -> a v) (\ v -> b v)
\end{verbatim}}

\subsection{Example in Section \ref{heuristic}}
The following is the input file for FCR. 

{
\begin{verbatim}
Ka : A x <= A (B x)
Kb : B x <= A x

g : forall a b x .
      (forall p y . p (a (b y)) => p (a y)) =>
      (forall p y . p (a y) => p (b y)) => a x
       
g a b = a (g (\ v . a (b v)) (\ v . a v))

h : A x 
h = g (\ v . Ka v) Kb

step h 20
\end{verbatim}
}

\noindent We use the alternative notation \texttt{A x <= A (B x)} to represent the rewrite rule
from \texttt{A x} to \texttt{A (B x)}, it will be translated to its Leibniz representation
by FCR. And \texttt{step h 20} is a command telling FCR
to output the 20th first-order term in the reduction \texttt{h} began with term \texttt{A x}. The following is the output information. 

{
\begin{verbatim}
rewrite rules
Ka : A x <= A (B x)
Kb : B x <= A x
kinds
A : * => *
B : * => *
axioms
Ka : forall p x . p (A (B x)) => p (A x)
Kb : forall p x . p (A x) => p (B x)
proof declarations
g : forall a b x .
      (forall p y . p (a (b y)) => p (a y))
      =>
        (forall p y . p (a y) => p (b y)) => a x =
\ a b . a (g (\ v . a (b v)) (\ v . a v))
h : A x =
g (\ v . Ka v) Kb
lemmas
h : A x =
  g (\ m1' . A m1') (\ m1' . B m1') x
    (\ p3' y4' (v : p3' (A (B y4'))) . Ka (\ m1' . p3' m1') y4' v)
    Kb
g : forall a b x .
      (forall p y . p (a (b y)) => p (a y))
      =>
        (forall p y . p (a y) => p (b y)) => a x =
  \ a0'
    b1'
    x2'
    (a : forall p y . p (a0' (b1' y)) => p (a0' y))
    (b : forall p y . p (a0' y) => p (b1' y)) .
      a (\ x1' . x1') x2'
        (g (\ m1' . a0' (b1' m1')) (\ m1' . a0' m1') x2'
           (\ p8' y9' (v : p8' (a0' (b1' (a0' y9')))) .
                a (\ m1' . p8' m1') (b1' y9')
                  (b (\ m1' . p8' (a0' (b1' m1'))) y9' v))
           (\ p14' y15' (v : p14' (a0' (b1' y15'))) .
                a (\ m1' . p14' m1') y15' v))
steps
step h 20
automated proof reconstruction success!
steps results
A (B (A (A (B (A (B (A (A (B (A (A (B x))))))))))))
\end{verbatim}
}

\noindent We can check that the term \texttt{A (B (A (A (B (A (B (A (A (B (A (A (B x))))))))))))}
represents the string we obtain in the very end of the string reduction trace in Section \ref{heuristic}.  
Note that this term is obtained directly from the unfolding of the reduction trace without invoking any term rewriting reduction.

\section{Solving the Scope Problem in ERSM and the Soundness of ERSM}
\label{scope}

Due to lack of space, we did not explain nor discuss the soundness of ERSM in Section \ref{heuristic}. In fact, the ERSM is not sound in its current form due to a subtle scope problem. We will
show how to solve this soundness problem in this section. To explain the scope problem, let us consider the following two formulas.

\

{
\noindent (I) \texttt{forall p x y . p (G (F Z x (S y)) (F x y (S (S Z)))) => p (F Z (S x) y)}

\

\noindent (II) \texttt{forall p x y . p (qa (F Z x (S y))) => p (F Z (S x) y)}
}

\

\noindent It may appear that these two formulas are second-orderly unifiable if we instantiate \texttt{qa} in (II) to
\texttt{\textbackslash m . G m (F x y (S (S Z)))}. 
But this instantiation assumes the variable \texttt{x, y} in \texttt{\textbackslash m . G m (F x y (S (S Z)))} can be automatically captured by the \texttt{forall} binder in (II), this is
not a correct assumption.
In fact (I) and (II) are not unifiable, this kind of problem is called
\textit{scope problem} by Dowek \cite[Section 5]{dowek2001}. %% As the existential variables \texttt{qa} may
%% be shared by other formulas, instantiating it with \texttt{\textbackslash m . G m (F x y (S (S Z)))} will bring \texttt{x, y} out of their current scope in (II).

The solution of the scope problem is conceptually simple, i.e. we just need to prevent the
instantiation of the existential variables when there is such a scope problem. However, to
implement this solution within the ERSM framework requires some efforts.

We works with \textit{idempotent} substitution, i.e. for a substitution $\sigma$,
we require that $\sigma \cdot \sigma = \sigma$. Idemptentness is easy to check, due to
the following property \cite{baader1999term}: $\sigma$ is idempotent iff $\mathrm{dom}(\sigma) \cap \mathrm{FV}(\mathrm{codom}(\sigma)) = \emptyset$. This requirement is needed in order to prove
the soundness theorem.

\begin{definition}
Let $L$ denote a list of variables. We define $y \sqsubset_L x$ if $L = L_1, y , L_2, x ,L_3$
  for some $L_1, L_2, L_3$. We define $\mathsf{scope}(L, \sigma)$ to be the conjunction of the following two predicates: (1) $\forall x\in \mathrm{dom}(\sigma) \cap L, \forall y \in \mathrm{FV}(\sigma {x}), y \sqsubset_L {x}$. (2) $\forall x \in \mathrm{dom}(\sigma) - L, \mathrm{FV}(\sigma x) \cap L = \emptyset$.
\end{definition}

Let $\Phi$ denotes a set of tuple $(L, \Gamma, e, T)$. We use $\sigma L$ to denote $L - \mathrm{dom}(\sigma)$ and we use $L + L'$ to mean appending $L, L'$.  
\begin{definition}
  \fbox{$\sigma\Gamma, \sigma \Phi$}

  \

  $\sigma \cdot = \cdot$

  $\sigma [\alpha : T, \Gamma] =  \alpha : \sigma T, \sigma \Gamma$

  $\sigma [\kappa : T, \Gamma] =  \kappa : \sigma T, \sigma \Gamma$
  
  $\sigma \{\} = \{\}$

  $\sigma \ \{(L, \Gamma, e, T), \Phi\} = \{(\sigma L, \sigma \Gamma, e, \sigma T), \sigma \Phi\}$, where $\mathsf{scope}(L, \sigma)$.
\end{definition}

Let $S$ be a set of variables, we write $\sigma/S = [ t/x \ | \ x \in (\mathrm{dom}(\sigma) - S)]$. 
\begin{definition}[ERSM with Scope Check]
    \fbox{$(\Phi, \sigma) \longrightarrow (\Phi', \sigma')$}

  \begin{enumerate}
  \item {\small $(\{(L, \Gamma, (\kappa | \alpha) \ e_1 \ ... \ e_n, A), \Phi \}, \sigma) \longrightarrow_a(\{(L', \sigma''\Gamma,e_1, \sigma' T_1), ..., (L', \sigma''\Gamma, e_1, \sigma' T_n), \sigma'' \Phi\},  \sigma'' \cdot \sigma)$}

    if $\kappa | \alpha : \forall x_1. ... \forall x_m. T_1,..., T_n \Rightarrow B \in \Gamma$ with $B \mapsto_{\sigma'} A$. Moreover, $\sigma'' = \sigma'/\{x_1,..., x_m\}$, $\mathsf{scope}(L, \sigma'')$ and $L' = \sigma'' L + [x_i \ | \ x_i \notin \mathrm{FV}(B), 1 \leq i \leq m]$. 

  %% \item $(\{(\Gamma, \alpha, A), \Phi\}, \sigma) \longrightarrow_b (\{\sigma' \Phi\}, \sigma' \cdot \sigma) $, if $\alpha : B \in \Gamma$, with $B \mapsto_{\sigma'} A$ and $\mathsf{scope}(\Gamma, \sigma')$. 

  \item $(\{(L, \Gamma, \lambda \alpha_1. ... \lambda \alpha_n. e, T_1, ..., T_n \Rightarrow A), \Phi\}, \sigma) \longrightarrow_i (\{(L, [\Gamma,\alpha_1 : T_1, ..., \alpha_n : T_n], e, A), \Phi \}, \sigma)$.

      \item $(\{(L, \Gamma, e, \forall x_1 ... \forall x_n. T), \Phi\}, \sigma) \longrightarrow_{\forall} (\{([L, x_1, ..., x_n], \Gamma, e, T), \Phi \}, \sigma)$.

  \item $(\{(L, \Gamma, \mu \alpha. e, T), \Phi\}, \sigma) \longrightarrow_c (\{(L, [\Gamma, \alpha : T], e, T), \Phi \}, \sigma)$.

      \end{enumerate}
\end{definition}
We can see if we eliminate $L$ and $\mathsf{scope}(L, \sigma)$, we can obtain ERSM described in
the paper. 
\begin{lemma}
  \label{type:subst}
    If $\Gamma \vdash e : T$, then $\sigma \Gamma \vdash \sigma e : \sigma T$.
\end{lemma}

If $S$ is a set of variables, we define $\sigma S := \{ \sigma x | \ x \in S\}$. Moreover, we
extend $\mathrm{FV}$ function to obtain all the free variables of a set of terms. Note that
all the
substitutions are \textit{idempotent} and \textit{disjoint} , i.e.
$\mathrm{FV}(\mathrm{codom}(\sigma)) \cap \mathrm{dom}(\sigma) = \emptyset$ for any $\sigma$ and $\mathrm{dom}(\sigma_1) \cap \mathrm{dom}(\sigma_2) = \emptyset, $ for any $\sigma_1, \sigma_2$. 

\begin{lemma}[Scope Check Composition]
  \label{compose}
  Suppose $\mathrm{FV}(\mathrm{codom}(\sigma_2)) \cap \mathrm{dom}(\sigma_1) = \emptyset$. If $\mathsf{Scope}(L, \sigma_1)$ and $\mathsf{Scope}(\sigma_1L + L', \sigma_2)$
  for some fresh $L'$, then $\mathsf{Scope}(L, \sigma_2 \cdot \sigma_1)$.
\end{lemma}
\begin{proof}
  \begin{itemize}
  \item Case $y \in \mathrm{dom}(\sigma_2 \cdot \sigma_1) - L$.

    We need to show
    $\mathrm{FV}(\sigma_2 \sigma_1 y) \cap L = \emptyset$, i.e. $\mathrm{FV}(\sigma_2 (\mathrm{FV}(\sigma_1 y))) \cap L = \emptyset$. We know that $\mathrm{dom}(\sigma_2 \cdot \sigma_1) = \mathrm{dom}(\sigma_2) \uplus \mathrm{dom}(\sigma_1)$. Suppose $y \in \mathrm{dom}(\sigma_1)$,
    we know that $\mathrm{FV}(\sigma_1 y) \cap L = \emptyset$. For any $z \in \mathrm{FV}(\sigma_1 y) \cap \mathrm{dom}(\sigma_2)$, we have $\mathrm{FV}(\sigma_2 z) \cap (\sigma_1L + L') = \emptyset$, which implies  $\mathrm{FV}(\sigma_2 z) \cap L = \emptyset$. For any $z \in \mathrm{FV}(\sigma_1 y) - \mathrm{dom}(\sigma_2)$, we have $\mathrm{FV}(\sigma_2 z)= \{z\}, \{z\} \cap L = \emptyset$. Suppose $y \in \mathrm{dom}(\sigma_2)$, we need to show  $\mathrm{FV}(\sigma_2  y) \cap L = \emptyset$, this is the case since $\mathrm{FV}(\sigma_2 y) \cap (\sigma_1L + L') = \emptyset$ and $\mathrm{FV}(\mathrm{codom}(\sigma_2)) \cap \mathrm{dom}(\sigma_1) = \emptyset$.

  \item Case. $y \in \mathrm{dom}(\sigma_2 \cdot \sigma_1) \cap L$.

    We need to show
    for any $z \in \mathrm{FV}(\sigma_2 (\mathrm{FV}(\sigma_1 y))) \cap L$, $z \sqsubset_{L} y$.
    Let $x \in \mathrm{FV}(\sigma_1 y)$, we just need to show for any $z \in \mathrm{FV}(\sigma_2 x) \cap L$, $z \sqsubset_{L} y$. Suppose $x \notin \mathrm{dom}(\sigma_2)$. Then $\mathrm{FV}(\sigma_2 x)= \{x\}$. So $x \sqsubset_{L} y$ if $x \in L$. 
Suppose $x \in \mathrm{dom}(\sigma_2) \cap L$, we know that $(\mathrm{FV}(\sigma_2 x) \cap (\sigma_1 L+L')) \sqsubset_{\sigma_1 L+L'} x$. Since $z \in \mathrm{FV}(\sigma_2 x) \cap L$ implies $z \in \mathrm{FV}(\sigma_2 x) \cap (\sigma_1 L+L')$, we have
    $z \sqsubset_{\sigma_1 L+L'} x \sqsubset_{L} y$. Since $x \notin L'$ and
    $x \notin \mathrm{dom}(\sigma_1)$, we have $z \sqsubset_{L} x \sqsubset_{L} y$. Suppose $x \in \mathrm{dom}(\sigma_2) - L$, then $x \in \mathrm{dom}(\sigma_2) - (\sigma_1L + L')$, thus $\mathrm{FV}(\sigma_2 x) \cap (\sigma_1L + L') = \emptyset$, which implies $\mathrm{FV}(\sigma_2 x) \cap L = \emptyset$.

    Suppose $y \in \mathrm{dom}(\sigma_2)$, we just need to show
    for any $z \in \mathrm{FV}(\sigma_2 y) \cap L$, $z \sqsubset_{L} y$. Since $z \notin \mathrm{dom}(\sigma_1)$, we have $z \in  \mathrm{FV}(\sigma_2 y) \cap (\sigma_1 L+L')$. Thus $z \sqsubset_{\sigma_1 L+L'} y$, which implies $z \sqsubset_{L} y$. 

  \end{itemize}
\end{proof}

\begin{lemma}[Scope Invariant]
  \label{sc:inv}

  \
  
  \begin{enumerate}
  \item
    If $(\{(L_1, \Gamma_1,e_1,T_1),..., (L_n, \Gamma_n,e_n,T_n)\},\sigma) \longrightarrow (\{(L_1', \Gamma_1',e_1',T_1'),..., (L_m', \Gamma_m',e_m',T_m')\}, \sigma'\cdot \sigma)$,
    %% and $\mathsf{Scope}(L_i, \sigma)$ for all $i$,
    then $\mathsf{Scope}(L_i, \sigma')$ for all $i$.
  \item If $(\{(L_1, \Gamma_1,e_1,T_1),..., (L_n, \Gamma_n,e_n,T_n)\},\sigma) \longrightarrow^* (\{(L_1', \Gamma_1',e_1',T_1'),..., (L_m', \Gamma_m',e_m',T_m')\}, \sigma'\cdot \sigma)$,
    %% and $\mathsf{Scope}(L_i, \sigma)$ for all $i$,
    then $\mathsf{Scope}(L_i, \sigma')$ for all $i$. 
  \end{enumerate}
\end{lemma}
\begin{proof}
  By Lemma \ref{compose} and induction.
\end{proof}

\begin{lemma}
  \label{app:ss}
  If $(\{(L_1, \Gamma_1,e_1,T_1),..., (L_n, \Gamma_n,e_n,T_n)\},\sigma) \longrightarrow^* (\emptyset, \sigma'\cdot \sigma)$ for some $\sigma'$, then $\sigma' \Gamma_i \vdash e_i' : \sigma' T_i$ and $|e_i'| = e_i$ for all $i$.

\end{lemma}
\begin{proof}
  By induction on the length of
  $(\{(L_1, \Gamma_1, e_1, T_1),..., (L_n, \Gamma_n, e_n, T_n)\}, \sigma)$ $\longrightarrow^* (\sigma'\cdot \sigma, \emptyset)$.
  \begin{itemize}
  \item Case $(\{(L, \Gamma, \alpha | \kappa, A)\}, \sigma) \longrightarrow_a (\emptyset, \sigma'' \cdot \sigma)$.
    
    In this case $\alpha | \kappa : \forall \underline{x}. B \in \Gamma$, $\sigma'' = \sigma'/\{\underline{x}\}$, $\mathsf{scope}(L, \sigma'')$ and $B \mapsto_{\sigma'} A$. By \textsc{Inst} rule and the idempotentness of $\sigma'$, we have $\sigma''\Gamma \vdash (\alpha | \kappa) \ (\sigma' \underline{x}) : \sigma' B \equiv \sigma'' \sigma' B = \sigma'' A$, where $|(\alpha | \kappa) \ (\sigma' \underline{x})| = \alpha | \kappa$.

  \item Case
    $(\{(L, \Gamma, (\alpha | \kappa)\ e''_1 \ ...  e_m'', A), (L_1, \Gamma_1, e_1, T_1), ...,
    (L_n, \Gamma_n, e_n, T_n)\}, \sigma) \longrightarrow_a $

    $(\{(L', \sigma'\Gamma, e_1'', \sigma_1 T_1'), ..., (L', \sigma' \Gamma, e_m'', \sigma_1 T_m'), (\sigma' L_1, \sigma' \Gamma_1, e_1, \sigma' T_1), ..., (\sigma' L_n, \sigma' \Gamma_n, e_n, \sigma' T_n)\}, \sigma' \cdot \sigma) \longrightarrow^* (\emptyset, \sigma'' \cdot \sigma' \cdot \sigma)$,

    where $\kappa | \alpha : \forall \underline{x}. T_1',..., T_n' \Rightarrow B \in \Gamma$ with $B \mapsto_{\sigma_1} A$, $\sigma' = \sigma_1/\{\underline{x}\}$, $\mathsf{scope}(L, \sigma')$ and $L' = \sigma' L + [x_i \ | \ x_i \notin \mathrm{FV}(B), 1 \leq i \leq m ]$. 
      
      By IH, we know that $\sigma'' \sigma'\Gamma' \vdash e_1''' : \sigma'' \sigma_1 T_1', ..., \sigma'' \sigma'\Gamma' \vdash e_m''' : \sigma'' \sigma_1 T_m', \sigma'' \sigma'\Gamma_1 \vdash e_1' : \sigma'' \sigma' T_1, ..., \sigma''\sigma'\Gamma_n \vdash e_n' : \sigma''\sigma' T_n$ and $|e_1'''| = e_1'',..., |e_m'''| = e_m'',|e_1'| = e_1, ..., |e_n'| = e_n$. We have $\sigma'' \sigma'\Gamma \vdash (\alpha | \kappa) \ (\sigma'' \sigma' \underline{x}): \sigma'' \sigma_1 T_1',..., \sigma'' \sigma_1 T_n' \Rightarrow \sigma''\sigma_1 B $. By 
      \textsc{Conv}, \textsc{App} and idempotentness, we have $\sigma'' \sigma' \Gamma \vdash (\alpha | \kappa) \ (\sigma'' \sigma' \underline{x})\ e_1'''\  ... \ e_m''' : \sigma''\sigma_1 B = \sigma'' \sigma' A$. Moreover, $|(\alpha | \kappa) \ (\sigma'' \sigma' \underline{x})\ e_1'''\  ... \ e_m'''| = (\alpha | \kappa) \ |e_1'''|\  ... \ |e_m'''| = (\alpha | \kappa) \ e_1''\  ... \ e_m''$.
      
      \item Case 
{
    $(\{(L, \Gamma, \lambda \alpha_1. ... \lambda \alpha_n. e, T_1, ..., T_n \Rightarrow A), (L_1, \Gamma_1, e_1, T_1), ..., (L_l, \Gamma_l, e_l, T_l) \}, \sigma) $
$\longrightarrow_i (\{(L, [\Gamma, \alpha_1 : T_1, ..., \alpha_n : T_n], e, A), (L_1, \Gamma_1, e_1, T_1), ..., (L_l, \Gamma_l, e_l, T_l) \}, \sigma) \longrightarrow^* (\emptyset, \sigma' \cdot \sigma)$}

By IH, we have $\sigma'\Gamma, \alpha_1 : \sigma'T_1, ..., \alpha_n : \sigma' T_n \vdash e' : \sigma' A, \sigma'\Gamma_1 \vdash e_1' : \sigma' T_1, ...,\sigma' \Gamma_l \vdash e_l' : \sigma' T_l$
with $|e'| = e, |e_1'| = e_1, ..., |e_l'| = e_l$. Thus by \textsc{Lam} rule, we have $\sigma'\Gamma \vdash \lambda \alpha_1 .... \lambda \alpha_n . e' : \sigma' T_1 
, ..., \sigma' T_n \Rightarrow \sigma' A$ and $| \lambda \alpha_1 .... \lambda \alpha_n . e'| = \lambda \alpha_1. ... \lambda \alpha_n. e$.

\item Case 
{
  $(\{(L, \Gamma,  e, \forall x_1. ...\forall x_m . T), (L_1, \Gamma_1, e_1, T_1), ..., (L_l, \Gamma_l, e_l, T_l) \}, \sigma)\longrightarrow_{\forall} $

  $(\{([L, x_1,..., x_m],\Gamma, e, T), (L_1, \Gamma_1, e_1, T_1), ..., (L_l, \Gamma_l, e_l, T_l) \}, \sigma) \longrightarrow^* (\emptyset, \sigma' \cdot \sigma) $}

By IH, we have $\sigma' \Gamma \vdash e' : \sigma' T, \sigma'  \Gamma_1 \vdash e_1' : \sigma'  T_1, ...,\sigma'  \Gamma_l \vdash e_l' : \sigma'  T_l$
with $|e'| = e, |e_1'| = e_1, ..., |e_l'| = e_l$. By Lemma \ref{sc:inv} (2), $\mathsf{scope}([L, x_1,..., x_m], \sigma')$. So $\mathrm{FV}(\mathrm{codom}(\sigma')) \cap \{x_1,..., x_m\} = \emptyset$. Thus by \textsc{Abs} rule, we have $\sigma'\Gamma \vdash \lambda x_1 .... \lambda x_m .  e' : \forall x_1 ... . \forall x_m . \sigma' T = \sigma' (\forall x_1 ... . \forall x_m . T)$ and $|\lambda x_1 .... \lambda x_m . e'| = e$.

\item Case  $(\{(L, \Gamma, \mu \alpha . e, T), (L_1, \Gamma_1, e_1, T_1), ..., (L_n, \Gamma_n, e_n, T_n)\}, \sigma) \longrightarrow_c $

  $(\{(L, [\Gamma, \alpha : T], e, T), (L_1, \Gamma_1, e_1, T_1), ..., (L_n, \Gamma_n, e_n, T_n)\}, \sigma) \longrightarrow^* (\emptyset, \sigma' \cdot \sigma)$
    
    By IH, we know that $\sigma' \Gamma, \alpha : \sigma' T \vdash e' : \sigma' T, \sigma'\Gamma_1 \vdash e_1' : \sigma' T_1, ..., \sigma' \Gamma_n \vdash e_n' : \sigma' T_n$ and $|e_i'| = e_i$ for all $i$. By \textsc{Mu} rule, we have $\sigma' \Gamma \vdash \mu \alpha . e' : \sigma' T$. Thus $|\mu \alpha. e'| = \mu \alpha. e$.

  \end{itemize}
\end{proof}

\begin{theorem}[Soundness of ERSM]
 If $(\{([], \Gamma, e, T)\}, \mathrm{id}) \longrightarrow^* (\emptyset,\sigma)$ and $\mathrm{FV}(\Gamma) = \mathrm{FV}(T) = \emptyset$, then $ \Gamma \vdash e' : T$ and $|e'| = e$. 
\end{theorem}
\begin{proof}
  By Lemma \ref{app:ss}.
\end{proof}

We now can understand the error message when we
try to type check the following declarations in FCR.

\begin{verbatim}
K : forall p x y . p (G (F Z x (S y)) (F x y (S (S Z)))) => p (F Z (S x) y)
    
K2 : forall qa . (forall p x y . p (qa (F Z x (S y))) => p (F Z (S x) y)) => B

h : B
h = K2 (\ c . K c)
\end{verbatim}

Note that type checking \texttt{h} will give a scope problem as (I) and (II)
above does not unify. FCR will print out the following message.

\begin{verbatim}
  scope error when matching [p1'] (qa0' (F Z [x2'] (S [y3'])))
  against [p1'] (G (F Z [x2'] (S [y3'])) (F [x2'] [y3'] (S (S Z))))
    when applying c : [p1'] (qa0' (F Z [x2'] (S [y3'])))
    when applying substitution [ qa0' : \ m1' .
                                          G m1' (F [x2'] [y3'] (S (S Z))) ]
    current variables list:
      qa0' p1' x2' y3'
    the current mixed proof term:
      K2 qa0'
        (\ p1' x2' y3' (c : [p1'] (qa0' (F Z [x2'] (S [y3'])))) .
             K (\ m1' . [p1'] m1') [x2'] [y3']
               ([p1'] (G (F Z [x2'] (S [y3'])) (F [x2'] [y3'] (S (S Z))))))
\end{verbatim}

\noindent The eigenvariables are the variables surrounded by brackets, and the substitution $[t/x]$
is represented as \texttt{[x : t]}. In this case the FCR will try to instantiate
the existential variable \texttt{qa0'} with \texttt{\textbackslash m1' . G m1' (F [x2'] [y3'] (S (S Z)))}. The $L$ is the \texttt{current variables list} for the $\mathsf{scope}$ function, we can see the substitution
will not pass the scope check. Moreover, we can inspect the mix proof term, we see that \texttt{qa0'} is
not in the scope of \texttt{[x2'], [y3']}. Thus the function \texttt{h} gives a typing error.

\section{Examples from Term Rewriting Literature}
\label{app:rewritingex}
We demonstrate 
how to use the prototype FCR to represent some nontrivial nonterminations in this section.
All of the examples in this section are from the existing term rewriting literature,
and we will focus on representing nonlooping nonterminating reductions.

The general idea of representing a nonterminating reduction trace is the following: we
need to see if the rule sequence can be generated by a corecursive function. Then
we will try to assign a type for the corecursive function. Most of the
efforts will be put on abstracting the right universal and existential type variables.
Obtaining the right type for the corecursive function usually requires interactions with
FCR and a good understanding of the type checking algorithm ERSM. 

% (3) demonstrate how FCR allows us to focus on describing the 
% nontermination trace and narrow down the right type annotation using the type checker
% problem.

\subsection{}
The following string rewriting system is from Endrullis and Zantema \cite{EndrullisZ15}, Example 29. 

\begin{center}
  $AL \to_1 LA$ \quad
  $RA \to_2 AR$\quad
  $BL \to_3 BR$\quad
  $RB \to_4 LAB$
\end{center}
Observe the following nonlooping nonterminating reduction: 

\begin{center}
  $\underline{BL}B \to_3 B\underline{RB} \to_4 \underline{BL}AB \to_3 B\underline{RA}B \to_2
  BA\underline{RB} \to_4 B\underline{AL}AB \to_1 \underline{BL}AAB \to_3 
  B\underline{RA}AB \to_2 BA\underline{RA}B 
  \to_2 BAA\underline{RB} \to_4 BA\underline{AL}AB \to_1 B\underline{AL}AAB \to_1
  \underline{BL}AAAB \to_3 ...  $
\end{center}

Observe that all the strings in the reduction can be described by the regular expression $BA^*(L|R)A^*B$. 
We focus on the 
rule sequence: $\underline{3}4\underline{32}41\underline{322}411.... $. The rule sequence can be generated
by the following corecursive function: $f \ a_1 \ a_2 \ a_3\ a_4 = a_3 \cdot a_4 \cdot (f \ a_1 \ a_2 \ (a_3 \cdot a_2)\ (a_4 \cdot a_1))$, i.e. $f \ 1\ 2\ 3\ 4$ gives the rule sequence. % Thus we need to obtain a 
% typing derivation for such function $f$ in $\mathbf{F}_{\omega\mu}$, note that this step is (highly) nontrivial because
% the type inference problem in $\mathbf{F}$ is undecidable. 

The term rewriting system corresponds to the above string rewriting system is the following. 
\begin{center}
  $A\ (L\ x) \to_1 L\ (A\ x)$ \quad
  $R\ (A\ x) \to_2 A\ (R\ x)$\quad
  $B\ (L\ x) \to_3 B\ (R\ x)$\quad
  $R\ (B\ x) \to_4 L\ (A\ (B\ x))$
\end{center}

The following is the type assignment for the function $f$, where the variable \texttt{r}
is an existential variable and will be instantiated by \texttt{(\textbackslash m1' . A (r2' m1'))} at the corecursive call of \texttt{f}.

{
\begin{verbatim}
K1 : A (L x) <= L (A x)
K2 : R (A x) <= A (R x)
K3 : B (L x) <= B (R x)
K4 : R (B x) <= L (A (B x))

f : forall p l r y . 
        (forall p x . p (l (A x)) => p (A (l x))) =>
        (forall p x . p (A (r x)) => p (r (A x))) =>
        (forall p x . p (B (r x)) => p (B (l x))) =>
        (forall p x . p (l (A (B x))) => p (r (B x))) => 
        p (B (l (B y)))

f a1 a2 a3 a4 = a3 (a4 (f (\ c . a1 c) 
                          (\ c . a2 c) 
                          (\ c . a3 (a2 c)) 
                          (\ c . a4 (a1 c))))

h : B (L (B y))
h = f K1 K2 (\ c . K3 c) K4
\end{verbatim}
  }

\subsection{}
The following string rewriting system is from Endrullis and Zantema \cite{EndrullisZ15}, Example 34. 

\begin{center}
  $ZL \to_1 LZ$ \quad
  $RZ \to_2 ZR$\quad
  $ZLL \to_3 ZLR$\quad
  $RRZ \to_4 LZRZ$
\end{center}
Observe the following nonlooping nonterminating reduction: 

\begin{center}
  $\underline{ZLL}ZZRZ \to_3 
  ZL\underline{RZ}ZRZ \to_2 ZLZ\underline{RZ}RZ 
  \to_2 ZLZZ\underline{RRZ} \to_4 ZLZ\underline{ZL}ZRZ \to_1 ZL\underline{ZL}ZZRZ \to_1
  \underline{ZLL}ZZZRZ \to_3 ZL\underline{RZ}ZZRZ \to_2 \cdot \to_2 \cdot \to_2 ZLZZZ\underline{RRZ} \to_4 
ZLZZ\underline{ZL}ZRZ \to_1 \cdot \to_1 \cdot \to_1 \underline{ZLL}ZZZZRZ \to_3 ... $
\end{center}
Observe the rule sequence: $32241132224111.... $. This rule sequence can be generated
by the following corecursive function: $f \ a_1 \ a_2 \ a_3\ a_4 = a_3 \cdot a_2 \cdot a_2 \cdot a_4 \cdot a_1 \cdot a_1 \cdot  (f \ a_1 \ a_2 \ (a_3 \cdot a_2)\ (a_4 \cdot a_1))$, i.e. $f \ 1\ 2\ 3\ 4$ gives the rule sequence. % Thus the nontrivial step is to obtain a
% typing derivation for such function $f$ in $\mathbf{F}_{\omega\mu}$. 

 The term rewriting system corresponds to the above string rewriting system is the following. 
\begin{center}
  $Z\ (L\ x) \to_1 L\ (Z \ x)$ \quad
  $R\ (Z\ x) \to_2 Z\ (R\ x)$\quad
  $Z\  (L\  (L\ x)) \to_3 Z\ (L\  (R\ x))$\quad
  $R\ (R\ (Z\ x)) \to_4 L\ (Z\  (R\  (Z\ x)))$
\end{center}

The following is the type that we assign to $f$. The existential variable \texttt{r} 
is instantiated by \texttt{(\textbackslash m1' . Z (r2' m1'))} at the corecursive call of \texttt{f}.

  {
\begin{verbatim}
K1 : Z (L x) <= L (Z x)
K2 : R (Z x) <= Z (R x)
K3 : Z (L (L x)) <= Z (L (R x))
K4 : R (R (Z x)) <= L (Z (R (Z x)))

f : forall p l r y . 
       (forall p x . p (l (Z x)) => p (Z (l x))) =>
       (forall p x . p (Z (r x)) => p (r (Z x))) =>
       (forall p x . p (Z (L (r x))) => p (Z (L (l x)))) =>                    
       (forall p x . p (l (Z (R (Z x)))) => p (r (R (Z x)))) => 
       p (Z (L (l (Z (Z (R (Z y)))))))

f a1 a2 a3 a4 = a3 (a2 (a2 (a4 (a1  (a1  (f (\ c . a1 c) 
                                            (\ c . a2 c) 
                                            (\ c . a3 (a2 c)) 
                                            (\ c . a4 (a1 c)))))))) 

h : (Z (L (L (Z (Z (R (Z y)))))))
h = f K1 K2 (\ c . K3 c) K4
\end{verbatim}
  }
  
\subsection{}

The following string rewriting system is from Endrullis and Zantema \cite{EndrullisZ15}, Example 33. 

\begin{center}
  $AAL \to_1 LAA$ \quad
  $RA \to_2 AR$\quad
  $BL \to_3 BR$\quad
  $RB \to_4 LAB$\quad
  $RB \to_5 ALB$
\end{center}
Observe the following nonlooping nonterminating reduction: 

\begin{center}
  $B\underline{RB} \to_4 
  \underline{BL}AB \to_3 B\underline{RA}B 
  \to_2 BA\underline{RB} \to_5 B\underline{AAL}B \to_1 \underline{BL}AAB \to_3
  B\underline{RA}AB \to_2 \cdot \to_2 BAA\underline{RB} \to_4 
 B\underline{AAL}AB \to_1  \underline{BL}AAAB \to_3 B\underline{RA}AAB \to_2 \cdot \to_2 \cdot \to_2 
 BAAA\underline{RB} \to_5 BAA\underline{AAL}B \to_1 \cdot \to_1 \underline{BL}AAAAB \to_3 B\underline{RA}AAAB
  \to_2 \cdot \to_2 \cdot \to_2 \cdot\to_2 BAAAA\underline{RB} \to_4 ... $
\end{center}
Observe the rule sequence:
$43251322, 41322251132222, 41132222251113222222.... $
This rule sequence can be generated
by the following corecursive function: $f \ a_1 \ a_2 \ a_3\ a_4 \ a_5 = a_4 \cdot a_3 \cdot a_2 \cdot a_5 
\cdot a_1 \cdot a_3 \cdot a_2 \cdot a_2 \cdot (f \ a_1 \ a_2 \ (a_3 \cdot a_2 \cdot a_2)\ (a_4 \cdot a_1)\ (a_5 \cdot a_1))$, i.e. $f \ 1\ 2\ 3\ 4\ 5$ gives the rule sequence. 

 The term rewriting system corresponds to the above string rewriting system is the following. 
\begin{center}
  $A \ (A\ (L\ x)) \to_1 L\ (A\ (A\ x))$ \quad
  $R\ (A\ x) \to_2 A\ (R\ x)$\quad
  $B\ (L\ x) \to_3 B\ (R\ x)$\quad
  $R\ (B\ x) \to_4 L\ (A\ (B\ x))$\quad
  $R\ (B\ x) \to_5 A\ (L\ (B\ x))$
\end{center}

We assign a type for $f$ in the following. The existential variable \texttt{l} is
instantiated with \texttt{(\textbackslash m1' . l1' (A (A m1')))} at the corecursive call of \texttt{f}.

\begin{verbatim}
K1 : A (A (L x)) <= L (A (A x))
K2 : R (A x) <= A (R x)
K3 : B (L x) <= B (R x)
K4 : R (B x) <= L (A (B x))
K5 : R (B x) <= A (L (B x))

f : forall p l r y . 
       (forall p x . p (l (A (A x))) => p (A (A (l x)))) => 
       (forall p x . p (A (r x)) => p (r (A x))) =>
       (forall p x . p (B (r x)) => p (B (l x))) =>
       (forall p x . p (l (A (B x))) => p (r (B x))) => 
       (forall p x . p (A (l (B x))) => p (r (B x))) => 
       p (B (r (B y)))

f a1 a2 a3 a4 a5 =
  a4 (a3 (a2 (a5 (a1 (a3 (a2 (a2 (f (\ c . a1 c)
                                    (\ c . a2 c) 
                                    (\ c . a3 (a2 (a2 c))) 
                                    (\ c . a4 (a1 c)) 
                                    (\ c . a5 (a1 c))))))))))

h : B (R (B y))
h = f K1 K2 K3 (\ c . K4 c) K5
\end{verbatim}

\subsection{}
\label{dum-elim}
  Consider the following rewriting system (from Zantema and Geser \cite{zantema1996}) :

  \begin{center}
    $F\ Z \ (S\ x)\ y \to_a F\ Z \ x \ (S \ y)$

    $F\ Z \ (S\ x)\ y \to_b F\ x \ y\ (S \ (S\ Z))$
  \end{center}
  
\noindent Observe the following nonlooping reduction trace.

\begin{center}
{  $F\ Z \ (S\ Z)\ (S\ Z) \to_b F\ Z \ (S\ Z)\ (S \ (S\ Z)) \to_b F\ Z
  \ (S \ (S\ Z)) \ (S \ (S\ Z)) \to_a F\ Z \ (S \ Z) \ (S \ (S\ (S\
  Z))) \to_b F\ Z \ (S \ (S\ (S\ Z)))\ (S\ (S\ Z)) \to_a ...$}
\end{center}

Note that the rule sequence for this reduction is: bbabaabaaab..... 
The nontermination can only be observed via
the full reduction tree. 
The following partial reduction tree produced by FCR is an infinite binary tree structure with each branch finite (by issuing command \texttt{:full 6 (F Z (S Z) (S Z))} to FCR). Each node is a triple (e.g. \texttt{[], B, F Z (S Z) (S (S Z))}), the first element denotes the redex position of the parent (which is a list of number, but all of them are at root position, hence \texttt{[]}), second element denotes the label of the rewrite rule applied, the third element denotes
the contractum.

  {
\begin{verbatim}
 [], _, F Z (S Z) (S Z)
|
+- [], B, F Z (S Z) (S (S Z))
|  |
|  +- [], B, F Z (S (S Z)) (S (S Z))
|  |  |
|  |  +- [], B, F (S Z) (S (S Z)) (S (S Z))
|  |  |
|  |  `- [], A, F Z (S Z) (S (S (S Z)))
|  |     |
|  |     +- [], B, F Z (S (S (S Z))) (S (S Z))
|  |     |  |
|  |     |  +- [], B, F (S (S Z)) (S (S Z)) (S (S Z))
|  |     |  |
|  |     |  `- [], A, F Z (S (S Z)) (S (S (S Z)))
|  |     |     |
|  |     |     +- [], B, F (S Z) (S (S (S Z))) (S (S Z))
|  |     |     |
|  |     |     `- [], A, F Z (S Z) (S (S (S (S Z))))
|  |     |
|  |     `- [], A, F Z Z (S (S (S (S Z))))
|  |
|  `- [], A, F Z Z (S (S (S Z)))
|
`- [], A, F Z Z (S (S Z))
\end{verbatim}}

Note that the rule sequence can be described by the corecursive function 
$f\ a_1 \ a_2 = a_2\ (f\ (\lambda c . a_1\ c)\ (\lambda c . a_2\ (a_1\ c)))$. 
We assign a type for $f$ in the following. The universal type variable $f$ is instantiated
by \texttt{\textbackslash m1' m2' m3' . f1' m1' m2' (S m3')} at the corecursive call of function \texttt{f}. We observe \texttt{step h 7} gives \texttt{F Z (S Z) (S (S (S (S Z))))}, which is the reducible leaf at depth 6 in the reduction tree.

{
\begin{verbatim}
A : forall p x y . p (F Z x (S y)) => p (F Z (S x) y)
B : forall p x y . p (F x y (S (S Z))) => p (F Z (S x) y)

f : forall p f . (forall p x y . p (f Z x (S y)) => p (f Z (S x) y)) => 
                 (forall p y . p (f Z y (S (S Z))) => p (f Z (S Z) y)) => 
                 p (f Z (S Z) (S Z))
f a1 a2 = a2 (f (\ c . a1 c) (\ c . a2 (a1 c)))

h : F Z (S Z) (S Z)
h = f A (\ c . B c)
step h 7
\end{verbatim}
  }

\subsection{}
\dbend

Consider the following one rule rewriting system (from Zantema and Geser \cite{zantema1996}) : 
  
  \begin{center}
    $F\ Z \ (S\ x)\ y \to_K G\ (F\ Z \ x \ (S \ y)) \ (F\ x \ y\ (S \ (S\ Z)))$
  \end{center}

  Note that the rewrite system in Section \ref{dum-elim} is the \textit{dummy eliminated} version
  of this rewriting system. Issuing command \texttt{:inner 6 (F Z (S Z) (S Z))} to FCR, we obtain the following reduction trace.

\begin{verbatim}
the execution trace is:
 F Z (S Z) (S Z)
-K-> G (F Z Z (S (S Z))) (F Z (S Z) (S (S Z)))
-K-> G (F Z Z (S (S Z)))
       (G (F Z Z (S (S (S Z)))) (F Z (S (S Z)) (S (S Z))))
-K-> G (F Z Z (S (S Z)))
       (G (F Z Z (S (S (S Z))))
          (G (F Z (S Z) (S (S (S Z)))) (F (S Z) (S (S Z)) (S (S Z)))))
-K-> G (F Z Z (S (S Z)))
       (G (F Z Z (S (S (S Z))))
          (G (G (F Z Z (S (S (S (S Z))))) (F Z (S (S (S Z))) (S (S Z))))
             (F (S Z) (S (S Z)) (S (S Z)))))
-K-> G (F Z Z (S (S Z)))
       (G (F Z Z (S (S (S Z))))
          (G (G (F Z Z (S (S (S (S Z)))))
                (G (F Z (S (S Z)) (S (S (S Z))))
                   (F (S (S Z)) (S (S Z)) (S (S Z)))))
             (F (S Z) (S (S Z)) (S (S Z)))))
-K-> G (F Z Z (S (S Z)))
       (G (F Z Z (S (S (S Z))))
          (G (G (F Z Z (S (S (S (S Z)))))
                (G (G (F Z (S Z) (S (S (S (S Z)))))
                      (F (S Z) (S (S (S Z))) (S (S Z))))
                   (F (S (S Z)) (S (S Z)) (S (S Z)))))
             (F (S Z) (S (S Z)) (S (S Z)))))
\end{verbatim}

In this case the rule sequence is pretty simple, so we cannot learn much from
the rule sequence. But when we observe the redexes, the reduction appear to
have the same patterns as the one in Section \ref{dum-elim}. The dummy elimination
technique makes the reduction pattern explicit in the rule sequence, it inspires
us to arrive at the following representation.

\begin{verbatim}
K : F Z (S x) y <= G (F Z x (S y)) (F x y (S (S Z)))

f : forall p qa qb f .
        (forall p x y . p (qa (f Z x (S y)) x y) => p (f Z (S x) y)) => 
        (forall p y . p (qb (f Z y (S (S Z))) y) => p (f Z (S Z) y)) =>
         p (f Z (S Z) (S Z)) 
f a1 a2 = a2 (f (\ c . a1 c) (\ c . (a2 (a1 c))))

h : F Z (S Z) (S Z)
h = f (\ c . K c) (\ c . K c)

step h 7
\end{verbatim}

The function \texttt{f} follows the exact same pattern as in Section \ref{dum-elim}, but
its type reflect the two use case of the rule \texttt{K}, i.e. applying \texttt{K} to the left or right argument of \texttt{G}. For each case we use a existential variable 
to capture the resulting contexts. Note that the existential variable \texttt{qa} has arity
3 and the existential variable \texttt{qb} has arity 2. Let us observe the following
fully annotated \texttt{h} and \texttt{f} from FCR. Notice that the third argument for
\texttt{f} in the definition of \texttt{h} is \texttt{\textbackslash m1' m2' . G (F Z Z (S m2')) m1'} (the order of \texttt{m1'} and \texttt{m2'} is switched in the body). And the third
argument is \texttt{\textbackslash m1' m2' . qb2' (qa1' m1' m2' (S (S Z))) (S m2')} at the corecursive call of \texttt{f} in the definition of \texttt{f} (the variable \texttt{m2'} is duplicated).

{\small
\begin{verbatim}
lemmas
h : F Z (S Z) (S Z) =
  f (\ x1' . x1') (\ m1' m2' m3' . G m1' (F m2' m3' (S (S Z))))
    (\ m1' m2' . G (F Z Z (S m2')) m1')
    (\ m1' m2' m3' . F m1' m2' m3')
    (\ p4'
       x5'
       y6'
       (c : p4' (G (F Z x5' (S y6')) (F x5' y6' (S (S Z))))) .
         K (\ m1' . p4' m1') x5' y6' c)
    (\ p10' y11' (c : p10' (G (F Z Z (S y11')) (F Z y11' (S (S Z))))) .
         K (\ m1' . p10' m1') Z y11' c)
f : forall p qa qb f .
      (forall p x y . p (qa (f Z x (S y)) x y) => p (f Z (S x) y))
      =>
        (forall p y . p (qb (f Z y (S (S Z))) y) => p (f Z (S Z) y))
        =>
          p (f Z (S Z) (S Z)) =
  \ p0'
    qa1'
    qb2'
    f3'
    (a1 : forall p x y .
            p (qa1' (f3' Z x (S y)) x y) => p (f3' Z (S x) y))
    (a2 : forall p y .
            p (qb2' (f3' Z y (S (S Z))) y) => p (f3' Z (S Z) y)) .
      a2 (\ m1' . p0' m1') (S Z)
        (f (\ m1' . p0' (qb2' m1' (S Z)))
           (\ m1' m2' m3' . qa1' m1' m2' (S m3'))
           (\ m1' m2' . qb2' (qa1' m1' m2' (S (S Z))) (S m2'))
           (\ m1' m2' m3' . f3' m1' m2' (S m3'))
           (\ p10'
              x11'
              y12'
              (c : p10' (qa1' (f3' Z x11' (S (S y12'))) x11' (S y12'))) .
                a1 (\ m1' . p10' m1') x11' (S y12') c)
           (\ p16'
              y17'
              (c : p16'
                     (qb2' (qa1' (f3' Z y17' (S (S (S Z)))) y17' (S (S Z))) (S y17'))) .
                a2 (\ m1' . p16' m1') (S y17')
                  (a1 (\ m1' . p16' (qb2' m1' (S y17'))) y17' (S (S Z)) c)))
\end{verbatim}
}

\subsection{}
\dbend 

The following term rewriting system is adapted from a string rewriting system in~\cite{EndrullisZ15}(Section 7), no current automated termination checker can detect the nontermination for this example.  

\begin{center}
$Bl\ (B\ x) \to_1 B\ (Bl\ x)$

  $Bl\ (Cl\ (Dl\ x)) \to_2 B\ (Cl\ (D\ x))$

  $D\ (Dl\ x) \to_3 Dl\ (D\ x)$

  $Al\ (X\ x) \to_4 Al\ (Bl\ (Bl\ x))$

  $B\ (X\ x) \to_5 X\ (Cl\ (Y\ x))$

  $Bl\ (Cl\ (Dl\ x)) \to_6 X\ (Cl\ (Y\ x))$

  $Y\ (D\ x) \to_7 Dl\ (Y\ x)$

  $Y\ (El\ x) \to_8 Dl\ (Dl\ (El\ x))$  
\end{center}

\noindent Observe the following nonlooping reduction trace ($\to_{a,b}$ is a shorthand
for $\to_a \cdot \to_b$): 
\begin{center}
  
  $Al \ (Bl \ \underline{(Bl \ (Cl \ (Dl} \ (Dl \ (El \ x)))))) \to_2 Al \ (Bl \ (B \ (Cl \ (D \ (Dl \ (El \ x)))))) \to_{1,3} 
  Al \ (B \ \underline{(Bl \ (Cl \ (Dl} \ (D \ (El \ x)))))) \to_6
  Al \ (B \ (X \ (Cl \ (Y \ (D \ (El \ x)))))) \to_{5,7} \underline{Al \ (X} \ (Bl \ (Cl \ (Dl \ \underline{(Y \ (El} \ x))))))
  \to_{4,8} Al \ (Bl \ (Bl \ \underline{(Bl \ (Cl \ (Dl} \ (Dl \ (Dl \ (El \ x)))))))) \to_2 
  Al \ (Bl \ (Bl \ (B \ (Cl \ (D \ (Dl \ (Dl \ (El \ x)))))))) \to_{1,3,1,3} 
  Al \ (B \ (Bl \ \underline{(Bl \ (Cl \ (Dl} \ (Dl \ (D \ (El \ x)))))))) \to_2 Al \ (B \ (Bl \ (B \ (Cl \ (Dl \ (Dl \ (D \ (El \ x))))))))
  \to_{1,3} Al \ (B \ (B \ \underline{(Bl \ (Cl \ (Dl} \ (D \ (D \ (El \ x)))))))) \to_6 Al \ (B \ (B \ (X \ (Cl \ (X \ (D \ (D \ (El \ x))))))))
  \to_{5,7,5,7} \underline{Al \ (X} \ (Bl \ (Bl \ (Cl \ (Dl \ (Dl \ \underline{(Y \ (El} \ x))))))))  \to_{4,8} Al \ (Bl \ (Bl \ (Bl \ (Bl \ (Cl \ (Dl \ (Dl \ (Dl \ (Dl \ (El \ x)))))))))) \to ...$
  
\end{center}

\noindent The rewriting system admits reductions of the form: 
{ $Al \ (Bl^n \ (Cl \ (Dl^n \ (El \ x)))))) \to^* Al \ (Bl^{n+1} \ (Cl \ (Dl^{n+1} \ (El \ x))))))$} for any for every $n>1$. The rule sequence of the above reduction is the following:
{ $213, 657, 48, 21313, 213, 65757, 48, 2131313, 21313, 213, 6575757, 48, ...$}. We now
represent this rule sequence by the following corecursive function:

\begin{center}
  $f\ a_1 \ a_2 \ a_3 \ a_4 \ a_5 \ a_6 \ a_7 \ a_8 \ b = (b \cdot a_6 \cdot
  a_5 \cdot a_7\cdot a_4 \cdot a_8) \ (f \ a_1 \ (a_2 \cdot a_1 \cdot a_3) \ a_3 \
  a_4 \ a_5 \ (a_6 \cdot a_5\ \cdot a_7) \ a_7 \ a_8\ \ (a_2 \cdot a_1
  \cdot a_3 \cdot a_1 \cdot a_3 \cdot b ))$
\end{center}

\noindent Note that $f \ 1 \ 2 \ 3\ 4\ 5\ 6\ 7\ 8\ (2 \cdot 1 \cdot 3) $ generates the rule sequence above. The following is the type we assign for $f$. 

{
\begin{verbatim}
K1 : Bl (B x) <= B (Bl x)
K2 : Bl (Cl (Dl x)) <= B (Cl (D x))
K3 : D (Dl x) <= Dl (D x)
K4 : Al (X x) <= Al (Bl (Bl x))
K5 : B (X x) <= X (Bl x)
K6 : Bl (Cl (Dl x)) <= X (Cl (Y x))
K7 : Y (D x) <= Dl (Y x)
K8 : Y (El x) <= Dl (Dl (El x))

f : forall p0 c b d y . 
      (forall p x . p (B (Bl x)) => p (Bl (B x))) =>
      (forall p x . p (B ( c (D x))) => p (Bl ( c (Dl x)))) =>
      (forall p x . p (Dl (D x)) => p (D (Dl x))) =>
      (forall p x . p (Al (Bl (Bl x))) => p (Al (X x))) =>
      (forall p x . p (X (Bl x)) => p (B (X x))) =>
      (forall p x . p (X ( c (Y x))) => p ( b (Cl ( d x)))) =>
      (forall p x . p (Dl (Y x)) => p (Y (D x))) =>
      (forall p x . p (Dl (Dl (El x))) => p (Y (El x))) =>
      (forall p x . p (B ( b (Cl ( d (D x))))) => p (Bl (Bl ( c (Dl (Dl x)))))) =>
      p0 (Al (Bl (Bl ( c (Dl (Dl (El y)))))))

f a1 a2 a3 a4 a5 a6 a7 a8 b = 
   b (a6 (a5 (a7 (a4 (a8 (f a1 
                            (\ c1 . a2 (a1 (a3 c1))) 
                            a3 
                            a4 
                            a5 
                            (\ c1. a6 (a5 (a7 c1))) 
                            a7 
                            a8 
                            (\ c1 . a2 (a1 (a3 (a1 (a3 (b c1))))))))))))

h : (Al (Bl (Bl ( Cl (Dl (Dl (El y)))))))
h = f K1 K2 K3 K4 K5 K6 K7 K8 (\ c . K2 (K1 (K3 c))) 
\end{verbatim}
}

\noindent Note that the quantified variables \texttt{b,d} in the type of \texttt{f} are existential variables.
In the corecursive call
of \texttt{f}, the variable \texttt{c} will be instantiated with \texttt{(\textbackslash m1' . Bl (c1' (Dl m1')))}
, \texttt{b} will be instantiated with \texttt{(\textbackslash m1' . B (b2' m1'))} and \texttt{d} will be instantiated with \texttt{(\textbackslash m1' . d3' (D m1'))}.

\subsection{}
\dbend 

The following rewriting system is from Emmes et. al. \cite{emmes2012proving}, which according to them is outside the scope of the their nontermination detection techniques.

\begin{center}
  $G \ T\ T \ x \ (S\ y) \to_1 G\ (N\ x) \ (N\ y)\ (S\ x) \ (D \ (S\ y))$

  $N \ Z \to_2 T$
  
  $N \ (S\ x) \to_3 N\ x$
  
  $D\ Z \to_4 Z$
  
  $D\ (S\ x) \to_5 S\ (S\ (D\ x))$
\end{center}

\noindent Observe the following nonlooping nonterminating reduction trace for $G \ T\ T \ Z \ (S\ Z)$ (using left to right, inner-most reduction strategy).

\begin{center}
 $G\ T\ T\ Z\ (S\ Z) \to_1 G\ (N\ Z)\ (N\ Z) (S\ Z) (D\ (S\ Z))
\to_2 G\ T\ (N\ Z) (S\ Z) (D\ (S\ Z))
\to_2 G\ T\ T\ (S\ Z) (D\ (S\ Z))
\to_5 G\ T\ T\ (S\ Z) (S\ (S\ (D\ Z)))
\to_4 G\ T\ T\ (S\ Z) (S\ (S\ Z))
\to_1 G\ (N\ (S\ Z)) (N\ (S\ Z)) (S\ (S\ Z)) (D\ (S\ (S\ Z)))
\to_3 G\ (N\ Z) (N\ (S\ Z)) (S\ (S\ Z)) (D\ (S\ (S\ Z)))
\to_2 G\ T\ (N\ (S\ Z)) (S\ (S\ Z)) (D\ (S\ (S\ Z)))
\to_3 G\ T\ (N\ Z) (S\ (S\ Z)) (D\ (S\ (S\ Z)))
\to_2 G\ T\ T\ (S\ (S\ Z)) (D\ (S\ (S\ Z)))
\to_5 G\ T\ T\ (S\ (S\ Z)) (S\ (S\ (D\ (S\ Z))))
\to_5 G\ T\ T\ (S\ (S\ Z)) (S\ (S\ (S\ (S\ (D\ Z)))))
\to_4 G\ T\ T\ (S\ (S\ Z)) (S\ (S\ (S\ (S\ Z)))) ... $
\end{center}

\noindent The rule sequence is of the shape $1,22,54,1,3232,554,1,3323332,55554...$. This rule sequence can be represented by the following corecursive equation.

\begin{center}
  $f\ a_1\ a_2\ b_2 \ a_3 \ b_3 \ a_4 \ a_5 = a_1\ a_2 \ b_2\ a_5 \
  a_4 \ (f \ a_1 \ (a_3 \cdot a_2)\ (b_3 \cdot b_2)\ a_3 \ (b_3 \cdot
  b_3) \ a_4 \ (a_5 \cdot a_5))$
\end{center}

\noindent Note that $f\ 1\ 2\ 2\ 3 \ 3 \ 4 \ 5$ gives rise to the rule sequence.
The following is the type that we assign to $f$.

{
\begin{verbatim}
K1 : forall p x y . p (G (N x) (N y) (S x) (D (S y))) => p (G T T x (S y))
K2 : forall p . p T => p (N Z)
K3 : forall p x . p (N x) => p (N (S x))
K4 : forall p . p Z => p (D Z)
K5 : forall p x . p (S (S (D x))) => p (D (S x))
f : forall p g n1 n2 s . 
      (forall p x y . p (g (n1 x) (n2 y) (S x) (D (s y))) => p (g T T x (s y))) =>
      (forall p . p T => p (n1 Z)) => 
      (forall p . p T => p (n2 Z)) => 
      (forall p x . p (n1 x) => p (n1 (S x))) => 
      (forall p x . p (n2 x) => p (n2 (s x))) =>
      (forall p . p Z => p (D Z)) => 
      (forall p x . p (s (s (D x))) => p (D (s x))) => 
       p (g T T Z (s Z))

f a1 a2 b2 a3 b3 a4 a5 =  
  a1 (a2 (b2 (a5 (a4 (f (\ c . a1 c) 
                        (\ c . a3 (a2 c)) 
                        (\ c . (b3 (b2 c))) 
                        (\ c . a3 c)
                        (\ c . b3 (b3 c))) 
                        a4 
                        (\ c . a5 (a5 c))))))

h : G T T Z (S Z)
h = f (\ c . K1 c) K2 K2 K3 K3 K4 K5
\end{verbatim}
  
  }

\noindent Note that \texttt{n1, n2} in the type of \texttt{f} are existential variables. At the corecursive call of \texttt{f}, variable \texttt{g} is instantiated by 
\texttt{(\textbackslash m1' m2' m3' m4' . g1' m1' m2' (S m3') m4')}, variable \texttt{n1} is instantiated
by \texttt{(\textbackslash m1' . n12' (S m1'))}, variable \texttt{n2} is instantiated
by \texttt{(\textbackslash m1' . n23' (s4' m1'))}, variable \texttt{s} is instantiated by \texttt{(\textbackslash m1' . s4' (s4' m1'))}.

\end{document}